\documentclass[11pt]{article}
\pdfoutput=1

\usepackage[a4paper,margin=1in]{geometry}
\usepackage{amsmath,amssymb,amsthm,amsfonts,cite,color,enumitem,graphicx,microtype}
\usepackage[small]{caption}
\usepackage{thmtools,thm-restate}
\usepackage{hyperref}

\newtheorem{theorem}{Theorem}[section]
\newtheorem{lemma}[theorem]{Lemma}
\newtheorem{corollary}[theorem]{Corollary}

\theoremstyle{definition}
\newtheorem{definition}[theorem]{Definition}
\theoremstyle{remark}
\newtheorem{remark}[theorem]{Remark}

\definecolor{darkgreen}{rgb}{0,0.5,0}
\definecolor{darkblue}{rgb}{0,0,0.8}
\hypersetup{
    colorlinks=true,
    linkcolor=darkblue,
    citecolor=darkgreen,
    urlcolor=darkblue,
    pdftitle={Improved Distributed Degree Splitting and Edge Coloring},
    pdfauthor={Mohsen Ghaffari, Juho Hirvonen, Fabian Kuhn, Yannic Maus, Jukka Suomela, Jara Uitto}
}

\usepackage[capitalize, nameinlink]{cleveref}
\crefname{theorem}{Theorem}{Theorems}
\Crefname{lemma}{Lemma}{Lemmas}
\Crefname{observation}{Observation}{Observations}
\Crefname{corollary}{Corollary}{Corollaries}
\Crefname{remark}{Remark}{Remarks}

\newcommand{\hide}[1]{}

\newcommand{\Tso}{T_{\mathsf{so}}}

\newcommand{\LOCAL}{\mathsf{LOCAL}}
\newcommand{\PRAM}{\mathsf{PRAM}}
\newcommand{\eps}{\varepsilon}

\newcommand{\logStar}[1]{\log^{*} #1}

\DeclareMathOperator{\polylog}{polylog}
\newcommand{\R}{\mathbb{R}}

\newcommand{\bigO}{O}

\newcommand{\discr}{\delta}
\newcommand{\lftwo}{\left\lfloor}
\newcommand{\rftwo}{\right\rfloor_*}

\begin{document}

\begin{flushleft}

\vspace*{10mm}
{\huge\bf Improved Distributed Degree Splitting\\and Edge
  Coloring\footnote{A preliminary version of this paper appeared in
    the 31$\mathrm{st}$ International Symposium on Distributed
    Computing (DISC 2017) \cite{DISC17}.}\par}
\vspace{10mm}

\newcommand{\auth}[3]{\textbf{#1}\par#2\par#3\par\medskip}

\auth{Mohsen Ghaffari}
{ETH Zurich}
{ghaffari@inf.ethz.ch}
\auth{Juho Hirvonen\footnote{Supported by Ulla Tuominen Foundation.}}
{Aalto University and IRIF, CNRS, and University Paris Diderot}
{juho.hirvonen@aalto.fi}
\auth{Fabian Kuhn\footnote{Supported by ERC Grant No.\ 336495 (ACDC).}}
{University of Freiburg}
{kuhn@cs.uni-freiburg.de}
\auth{Yannic Maus\footnotemark[3]}
{University of Freiburg}
{yannic.maus@cs.uni-freiburg.de}
\auth{Jukka Suomela}
{Aalto University}
{jukka.suomela@aalto.fi}
\auth{Jara Uitto}
{ETH Zurich and University of Freiburg}
{jara.uitto@inf.ethz.ch}
\end{flushleft}

\vspace{5mm}

\paragraph{Abstract.}
The degree splitting problem requires coloring the edges of a graph red or blue such that each node has almost the same number of edges in each color, up to a small additive discrepancy. The directed variant of the problem requires orienting the edges such that each node has almost the same number of incoming and outgoing edges, again up to a small additive discrepancy.

We present deterministic distributed algorithms for both variants, which improve on their counterparts presented by Ghaffari and Su [SODA'17]: our algorithms are significantly simpler and faster, and have a much smaller discrepancy. This also leads to a faster and simpler deterministic algorithm for $(2+o(1))\Delta$-edge-coloring, improving on that of Ghaffari and Su.

\clearpage

%!TEX root = main.tex

\section{Introduction and Related Work}

In this work, we present improved distributed ($\LOCAL$ model) algorithms for the \emph{degree splitting problem}, and also use them to provide simpler and faster deterministic distributed algorithms for the classic and well-studied problem of \emph{edge coloring}.

\paragraph{\boldmath $\LOCAL$ Model.}
In the standard $\LOCAL$ model of distributed computing\cite{linial1987LOCAL, Peleg:2000}, the network is abstracted as an $n$-node undirected graph $G=(V, E)$, and each node is labeled with a unique $O(\log n)$-bit identifier. Communication happens in synchronous rounds of \emph{message passing}, where in each round each node can send a message to each of its neighbors. At the end of the algorithm each node should output its own part of the solution, e.g., the colors of its incident edges in the edge coloring problem. The time complexity of an algorithm is the number of synchronous rounds.

\paragraph{Degree Splitting Problems.}
The \emph{undirected degree splitting} problem seeks a partitioning of the graph edges $E$ into two parts so that the partition looks almost balanced around each node. Concretely, we should color each edge red or blue such that for each node, the difference between its number of red and blue edges is at most some small \emph{discrepancy} value $\kappa$. In other words, we want an assignment $q\colon E\rightarrow \{+1, -1\}$ such that for each node $v \in V$, we have
\[\textstyle\bigl|\sum_{e\in E(v)} q(e)\bigr|\leq \kappa,\]
where $E(v)$ denotes the edges incident on $v$. We want $\kappa$ to be as small as possible.

In the \emph{directed} variant of the {degree splitting} problem, we should orient all the edges such that for each node, the difference between its number of incoming and outgoing edges is at most a small discrepancy value $\kappa$.

\paragraph{Why Should One Care About Distributed Degree Splittings?}
On the one hand, degree splittings are natural tools for solving other problems with a \emph{divide-and-conquer} approach. For instance, consider the well-studied problem of edge coloring, and suppose that we are able to solve degree splitting efficiently with discrepancy $\kappa=O(1)$. We can then compute an edge coloring with ${(2+\eps)\Delta}$ colors, for any constant $\eps>0$; as usual, $\Delta$ is the maximum degree of the input graph $G=(V,E)$. For that, we recursively apply the degree splittings on $G$, each time reapplying it on each of the new colors, for a recursion of height $h=O(\log \eps\Delta)$. This way we partition $G$ in $2^{h}$ edge-disjoint graphs, each with maximum degree at most
\[\Delta'=\frac{\Delta}{2^{h}} + \sum_{i=1}^{h}\frac{\kappa}{2^{i}} \leq \frac{\Delta}{2^{h}} + \kappa = O(1/\eps).\]
We can then edge color each of these graphs with $2\Delta'-1$ colors, using standard algorithms (simultaneously in parallel for all graphs and with a separate color palette for each graph), hence obtaining an overall coloring for $G$ with $2^{h} \cdot (2\Delta'-1) \leq 2\Delta + 2^{h}\kappa = (2+\eps)\Delta$ colors. We explain the details of this relation, and the particular edge coloring algorithm that we obtain using our degree splitting algorithm, later in \Cref{crl:edgeColoring}.

On the other hand, degree splitting problems are interesting also on their own: they seem to be an elementary locally checkable labeling (LCL) problem\cite{naor1993can}, and yet, even on bounded degree graphs, their distributed complexity is highly non-trivial. In fact, they exhibit characteristics that are intrinsically different from those of the classic problems of the area, including maximal independent set, maximal matching, $\Delta+1$ vertex coloring, and $2\Delta-1$ edge coloring. All of these classic problems admit trivial sequential greedy algorithms, and they can also be solved very fast distributedly on bounded degree graphs, in $\Theta(\log^* n)$ rounds\cite{linial1987LOCAL}. In contrast, degree splittings constitute a middle ground in the complexity: even on bounded degree graphs, deterministic degree splitting requires $\Omega(\log n)$ rounds, as shown by Chang et al.~\cite{chang2016exponential}, and randomized degree splitting requires $\Omega(\log\log n)$ rounds, as shown by Brandt et al.~\cite{brandt2016lower}. These two lower bounds were presented for the \emph{sinkless orientation} problem, introduced by Brandt et al.~\cite{brandt2016lower}, which can be viewed as a very special case of directed degree splitting: In sinkless orientation, we should orient the edges so that each node of degree at least $d$, for some large enough constant $d$, has at least one outgoing edge. For this special case, both lower bounds are tight\cite{ghaffari17}.

\paragraph{What is Known?}
First, we discuss the existence of low-discrepancy degree splittings. Any graph admits an undirected degree splitting with discrepancy at most $2$. This is the best possible, as can be seen on a triangle. This low-discrepancy degree splitting can be viewed as a special case of a beautiful area called \emph{discrepancy theory} (see e.g. \cite{chazelle2000discrepancy} for a textbook coverage), which studies coloring the elements of a ground set red/blue so that each of a collection of given subsets has almost the same number of red and blue elements, up to a small additive discrepancy. For instance, by a seminal result of Beck and Fiala from 1981\cite{beck1981integer}, any hypergraph of rank $t$ (each hyperedge has at most $t$ vertices) admits a red/blue edge coloring with per-node discrepancy at most $2t-2$. See \cite{bukh2016improvement, bednarchak1997note} for some slightly stronger bounds, for large $t$. In the case of standard graphs, where $t=2$, the existence proof is straightforward: Add a dummy vertex and connect it to all odd-degree vertices.  Then, take an Eulerian tour, and color its edges red and blue in an alternating manner. In directed splitting, a discrepancy of $\kappa=1$ suffices, using the same Eulerian tour approach and orienting the edges along a traversal of this tour.

In the algorithmic world, Israeli and Shiloach~\cite{israeli1986improved} were the first to consider degree splittings. They used it to provide an efficient parallel ($\PRAM$ model) algorithm for maximal matching. This, and many other works in the $\PRAM$ model which later used degree splittings (e.g., \cite{karloff1987efficient}) relied on computing Eulerian tours, following the above scheme. Unfortunately, this idea cannot be used efficiently in the distributed setting, as an Eulerian tour is a non-local structure: finding and alternately coloring it needs $\Omega(n)$ rounds on a simple cycle.

Inspired by Israeli and Shiloach's method~\cite{israeli1986improved}, Hanckowiak et al.~\cite{hanckowiak01} were the first to study degree splittings in distributed algorithms. They used it to present the breakthrough result of a $\polylog(n)$-round deterministic distributed maximal matching, which was the first efficient deterministic algorithm for one of the classic problems. However, for that, they ended up having to relax the degree splitting problem in one crucial manner: they allowed a $\delta=1/\polylog n$ fraction of nodes to have arbitrary splits, with no guarantee on their balance. As explained by Czygrinow et al.~\cite{czygrinow2001coloring}, this relaxation ends up being quite harmful for edge coloring; without fixing that issue, it seems that one can get at best an $O(\Delta\log n)$-edge coloring.

Very recently, Ghaffari and Su\cite{ghaffari17} presented solutions for degree splitting without sacrificing any nodes, and used this to obtain the first $\polylog n$ round algorithm for ${(2+o(1))\Delta}$-edge coloring, improving on prior $\polylog(n)$-round algorithms that used more colors: the algorithm of Barenboim and Elkin~\cite{Barenboim:edge-coloring} for ${\Delta\cdot \exp(O(\frac{\log \Delta}{\log\log \Delta}))}$ colors, and the algorithm of Czygrinow et al.~\cite{czygrinow2001coloring} for $O(\Delta\log n)$ colors. The degree splitting algorithm of Ghaffari and Su\cite{ghaffari17} obtains a discrepancy $\kappa=\eps\Delta$ in $O((\Delta^2 \log^5 n)/\eps)$ rounds. Their method is based on iterations of flipping augmenting paths (somewhat similar in style to blocking flows in classic algorithms for the maximum flow problem\cite{dinitz2006dinitz}) but the process of deterministically and distributedly finding enough disjoint augmenting paths is quite complex. Furthermore, that part imposes a crucial limitation on the method: it cannot obtain a discrepancy better than $\Theta(\log n)$. As such, this algorithm does not provide any meaningful solution in graphs with degree $o(\log n)$.

\paragraph{Our Contributions.}
Our main result is a deterministic distributed algorithm for degree splitting that improves on the corresponding result of \cite{ghaffari17}. The new algorithm is (1) simpler, (2) faster, and (3) it gives a splitting with a much lower discrepancy.

\pagebreak
\begin{restatable}{theorem}{thmMainSplitting}\label{thm:mainSplitting}
For every $\eps>0$, there are deterministic $O\big(\eps^{-1}\cdot\log\eps^{-1}\cdot\big(\log\log\eps^{-1}\big)^{1.71}\cdot \log n\big)$-round distributed algorithms for computing directed and undirected degree splittings with the following properties:
  \begin{enumerate}[label=(\alph*)]
  \item For directed degree splitting, the discrepancy at each node $v$ of degree $d(v)$ is at most $\eps \cdot d(v) + 1$ if $d(v)$ is odd and at most $\eps\cdot  d(v) + 2$ if $d(v)$ is even.
  \item For undirected degree splitting, the discrepancy at each node $v$ of degree $d(v)$ is at most $\eps\cdot d(v) + 4$.
  \end{enumerate}
\end{restatable}

An important corollary of this splitting result is a faster and simpler algorithm for $(2+o(1))\Delta$-edge coloring, which improves on the corresponding result from \cite{ghaffari17}. The related proof is deferred to \Cref{sec:edgeColoring}.

\begin{restatable}{corollary}{crlEdgeColoring}\label{crl:edgeColoring}
For every $\eps>1/\log \Delta$, there is a deterministic distributed algorithm that computes a $(2+\eps)\Delta$-edge coloring in $\bigO\big(\log^2\Delta \cdot \eps^{-1} \cdot \log\log \Delta \cdot (\log\log\log\Delta)^{1.71} \cdot \log n\big)$ rounds.
\end{restatable}

This is significantly faster than the $\bigO(\log^{11} n/\eps^3)$-round algorithm of \cite{ghaffari17}. Subsequent and partly also in parallel to the work on the conference version of this paper, there has been further significant progress in the development of deterministic distributed edge coloring algorithms. This in particular includes the 
 first polylogarithmic-time deterministic $(2\Delta-1)$-edge coloring algorithm in \cite{FOCS17} by Fischer, Ghaffari, and Kuhn, which requires $\bigO(\log^7 n)$ rounds. This has afterwards been improved to $\bigO(\log^6 n)$ rounds by Ghaffari, Harris and Kuhn in \cite{FOCS18} and to $\bigO(\log^4 n)$ rounds by Harris in \cite{HarrisDerandomization}. In \cite{GKMU17}, Ghaffari, Kuhn, Maus and Uitto even go below the threshold of $2\Delta-1$ colors and provide deterministic polylogarithmic-time algorithms for $(1+\eps)\Delta$-edge coloring. The splitting result of the current paper plays an important role in the latter result: This splitting brings down the degree to a small value, with a negligible $(1+o(1))$ factor loss, and then those small degree graphs are colored efficiently.

\Cref{thm:mainSplitting} has another fascinating consequence. Assume that we have a graph in which all nodes have an odd degree. If $\eps<1/\Delta$, we get a directed degree splitting in which each node $v$ has outdegree either $\lfloor d(v)/2 \rfloor$ or $\lceil d(v)/2 \rceil$. Note that the number of nodes for which the outdegree is $\lfloor d(v)/2 \rfloor$ has to be exactly $n/2$. We therefore get an efficient distributed algorithm to exactly divide the number of nodes into two parts of equal size in any odd-degree graph. For bounded-degree graphs, the algorithm even runs in time $\bigO(\log n)$.

\paragraph{Our Method in a Nutshell.}
The main technical contribution is a distributed algorithm that partitions the edge set of a given graph in \emph{edge-disjoint short paths} such that each node is the start or end of at most $\delta$ paths. We call such a partition a \emph{path decomposition} and $\delta$ its \emph{degree} (cf. \Cref{fig:pathDecomp} for an illustration of a path decomposition). Now if we orient each path of a path decomposition with degree $\delta$ consistently, we obtain an orientation of discrepancy at most $\delta$. Moreover, such an orientation can be computed in time which is linear in the maximum path length. 

To study path decompositions in graph $G$, it is helpful to consider an auxiliary graph $H$ in which each edge $\{u,v\}$ represents a path from $u$ to $v$ in $G$; now $\delta$ is the maximum degree of graph~$H$. To construct a  low-degree path decomposition where $\delta$ is small, we can start with a trivial decomposition $H = G$, and then repeatedly join pairs of paths: we can replace the edges $\{u,v_1\}$ and $\{u,v_2\}$ in graph $H$ with an edge $\{v_1,v_2\}$, and hence make the degree of $u$ lower, at a cost of increasing the path lengths---this operation is called a \emph{contraction} here.

If each node $u$ simply picked arbitrarily some edges $\{u,v_1\}$ and $\{u,v_2\}$ to contract, this might result in long paths or cycles. The key idea is that we can use a \emph{high-outdegree orientation} to select a good set of edges to contract: Assume that we have an orientation in $H$ such that all nodes have outdegree at least $2k$. Then each node could select $k$ pairs of outgoing edges to contract; this would reduce the maximum degree of $H$ from $\delta$ to $\delta - 2k$ and only double the maximum length of a path. Also see the illustrations of this contracting process in \Cref{fig:contract,fig:contractOne}.

In essence, this idea makes it possible to \emph{amplify} the quality of an orientation algorithm: Given an algorithm $A$ that finds an orientation with a large (but not optimal) outdegree, we can apply $A$ repeatedly to reduce the maximum degree of $H$. This will result in a low-degree path decomposition of $G$, and hence also provide us with a well-balanced orientation in $G$.

\paragraph{Structure.}
The roadmap for this paper is as follows:
\begin{itemize}[noitemsep]
    \item \Cref{sec:shortPathsDecompositions}: Partitioning graphs in edge-disjoint short paths (main technical contribution).
    \item \Cref{sec:basecase}: Finding orientations in $3$-regular graphs (used in \Cref{sec:outdegtwo} and in \Cref{sec:mainsplitting}).
    \item \Cref{sec:outdegtwo}: Finding orientations in $5$-regular graphs (used in \Cref{sec:shortPathsDecompositions}).
    \item \Cref{sec:mainsplitting}: Proof of \Cref{thm:mainSplitting}.
    \item \Cref{sec:edgeColoring}: Proof of \Cref{crl:edgeColoring}.
    \item \Cref{sec:weak2orientation-lb}: A lower bound for orientations in even-degree graphs.
\end{itemize}
Here \Cref{sec:shortPathsDecompositions} is the most interesting part; \Cref{sec:basecase} and \Cref{sec:outdegtwo} deal with some corner cases that are needed in order to have tight constants for odd-degree graphs.

%!TEX root = main.tex

\section{Short Path Decompositions}\label[section]{sec:shortPathsDecompositions}

The basic building block of our approach is to find consistently oriented and short (length $\bigO(\Delta)$) paths in an oriented graph.
The first crucial observation is that an oriented path going through a node $v$ is ``good'' from the perspective of $v$ in the sense that it provides exactly one incoming and one outgoing edge to $v$.
Another important feature is that flipping a consistently oriented path does not increase the discrepancy between incoming and outgoing edges for any non-endpoint node along the path.
Following these observations, we recursively decompose a graph into a set of short paths, and merge the paths to ensure that every node is at the end of only a few paths. If a node $v$ is at the end of $\delta(v)$ paths an arbitrary orientation of these paths will provide a split with discrepancy at most $\delta(v)$ for $v$.

The recursive graph operations may turn graphs into multigraphs with self-loops. Thus throughout the chapter a multigraph is allowed to have self-loops and the nodes of a path $v_1, \ldots, v_k$ do not need to be distinct; however, a path can contain each edge at most once. A self-loop at a node $v$ contributes two to the degree of $v$.

\subsection{Orientations and Edge Contractions}

The core concept to merge many paths in parallel in one step of the aforementioned recursion is given by the concept of weak $k(v)$-orientations. We begin by extending and adapting prior work \cite{ghaffari17} on weak orientations to our needs.
\begin{definition}\label{def:sinkless}
A \emph{weak $k(v)$-orientation} of a multigraph $G=(V,E)$ is an orientation of the edges $E$ such that each node $v\in V$ has outdegree at least $k(v)$.
\end{definition}
Note that a weak $1$-orientation is a \emph{sinkless orientation}. By earlier work, it is known that a weak $1$-orientation can be found in time $\bigO(\log n)$ in simple graphs of minimum degree at least three.
\begin{lemma}[Sinkless Orientation, \cite{ghaffari17}]
\label[lemma]{lemma:weak1}
A weak $1$-orientation can be computed by a deterministic algorithm in $\bigO(\log n)$ rounds in simple graphs with minimum degree $3$ (and by a randomized algorithm in $\bigO(\log \log n)$ rounds in the same setting).
\end{lemma}
In our proofs, we may face multigraphs with multiple self-loops and with nodes of degree less than three and thus, we need a slightly modified version of this result.
\begin{corollary}[Sinkless Orientation, \cite{ghaffari17}]
	\label[corollary]{lemma:weakmulti}
	Let $G = (V, E)$ be a multigraph and $W \subseteq V$ a subset of nodes with degree at least three. Then, there is a deterministic algorithm that finds an orientation of the edges such that every node in $W$ has outdegree of at least one and runs in $\bigO(\log n)$ rounds (and a randomized algorithm that runs in $\bigO(\log \log n)$ rounds).
\end{corollary}
\begin{proof}
	For every multi-edge, both endpoints pick one edge and orient it outwards, ties broken arbitrarily. For every self-loop, the node will orient it arbitrarily. This way, every node with an incident multi-edge or self-loop will have an outgoing edge.

	From here on, let us ignore the multi-edges and self-loops and focus on the simple graph $H$ remaining after removing the multi-edges.
	For every node $v$ with degree at most two in $H$, we connect $v$ to $3 - d(v)$ copies of the following gadget $U$.
	The set of nodes of $U = \{ u_1, u_2, u_3, u_4, u_5 \}$ is connected as a cycle.
	Furthermore, we add edges $\{u_2, u_4\}$ and $\{u_3, u_5\}$ to the gadget and connect $u_1$ to $v$.
	This way, the gadget is $3$-regular.

	In the simple graph constructed by adding these gadgets, we run the algorithm of \Cref{lemma:weak1}.
	Thus, any node of degree at least three in the original graph that was not initially adjacent to a multi-edge or self-loop gets an outgoing edge.
	Since we know that every node incident to a multi-edge or self-loop in $G$ also has an outgoing edge, the claim follows.
\end{proof}

The sinkless orientation algorithm from \Cref{lemma:weakmulti} immediately leads to an algorithm which finds a weak $\lfloor d(v)/3\rfloor$-orientation in multigraphs in time $\bigO(\log n)$.
\begin{lemma}[Weak $\lfloor d(v) /3\rfloor$-Orientation]
\label[lemma]{lemma:weakDelta3}
There is a deterministic algorithm that finds a weak $\lfloor d(v) /3\rfloor$-orientation in time $\bigO(\log n)$ in multigraphs.
\end{lemma}
\begin{proof}
Partition node $v$ into $\lceil d(v)/3\rceil$ nodes and split its adjacent edges among them such that $\lfloor d(v)/3\rfloor$ nodes have exactly three adjacent edges each and the remaining node, if any, has $d(v) \bmod 3$ adjacent edges. Note that the partitioning of a node into several nodes may cause self-loops to go between two different copies of the same node. Then, use the algorithm from \Cref{lemma:weakmulti} to compute a weak $1$-orientation of the resulting multigraph where degree two or degree one nodes do not have any outdegree requirements. If we undo the partition but keep the orientation of the edges we have a weak $\lfloor d(v)/3 \rfloor$-orientation of the original multigraph.
\end{proof}

The concept of weak orientations can be extended to both indegrees and outdegrees.
\begin{definition}
A \emph{strong $k(v)$-orientation} of a multigraph $G=(V,E)$ is an orientation of the edges $E$ such that each node $v\in V$ has both indegree and outdegree at least $k(v)$.
\end{definition}

The techniques in this section need orientations in which nodes have at least two outgoing edges. \Cref{lemma:weakDelta3} provides such orientations for nodes of degree at least six; but for nodes of smaller degree it guarantees only one outgoing edge. It is impossible to improve this for nodes with degree smaller than five in time $o(n)$ (cf. \Cref{thm:weak2orientation-lb}). But we obtain the following result for the nodes with degree five. Its proof relies on different techniques than the techniques in this section, and therefore it is deferred to \Cref{sec:outdegtwo}.

\begin{restatable}[Outdegree 2]{lemma}{lemmaFirstoutdegtwo}\label{lemma:firstoutdeg-2}
The following problem can be solved in time $\bigO(\log n)$ with deterministic algorithms and $\bigO(\log \log n)$ with randomized algorithms:
given any multigraph, find an orientation such that all nodes of degree at least $5$ have outdegree at least $2$.
\end{restatable}

\subsection{Path Decompositions}

We now introduce the concept of a path decomposition.
The decomposition proves to be a strong tool due to the fact that it can be turned into a strong orientation (cf.~\Cref{lemma:fromPathDecompToStrongOrient}).
\begin{definition}[Path Decomposition]
Given a multigraph $G=(V,E)$, a positive integer~$\lambda$, and a function $\delta\colon V\rightarrow \R_{\geq 0}$, we call a partition $\mathcal{P}$ of the edges $E$ into edge-disjoint paths $P_1,\ldots, P_{\rho}$ a \emph{$(\delta,\lambda)$-path decomposition} if
\begin{itemize}[noitemsep]
\item for every $v\in V$ there are at most $\delta(v)$ paths that start or end in $v$,
\item each path $P_i$ is of length at most $\lambda$.
\end{itemize}
For each path decomposition $\mathcal{P}$, we define the multigraph $G(\mathcal{P})$ as follows: the vertex set of $G(\mathcal{P})$ is $V$, and there is an edge between two nodes $u,v\in V$ if $\mathcal{P}$ has a path which starts at $u$ and ends at $v$ or vice versa.
The \emph{degree of $v$ in $\mathcal{P}$} is defined to be its degree in $G(\mathcal{P})$ and the \emph{maximum degree  of the path decomposition $\mathcal{P}$} is the maximum degree of $G(\mathcal{P})$.
\end{definition}
\begin{figure}
\centering
\includegraphics[width=0.6\textwidth]{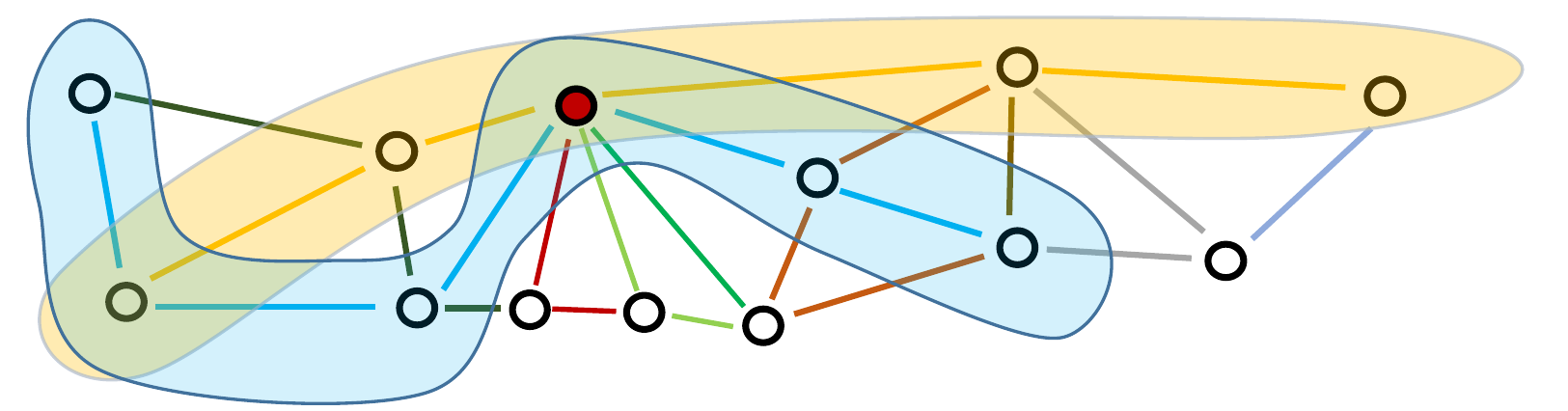}
\caption{The longest path in the given path decomposition has length five and there are three paths that start/end at the red node.}
\label{fig:pathDecomp}
\end{figure}

Notice that $\delta(v)$ is an upper bound on the degree of $v$ in $\mathcal{P}$ and $\max_{v\in V}\delta(v)$ is an upper bound on the maximum degree of the path decomposition. Note that $d_G(v)-d_{G(\mathcal{P})}(v)$ is always even. If the function in some path decomposition $(\delta,\lambda)$-path decomposition satisfies $\delta(v)=a$ for some $a$ we also speak of a $(a,\lambda)$-path decomposition.  See \Cref{fig:pathDecomp} for an illustration of a path decomposition and its parameters.
To make proofs more to the point instead of getting lost in notation, we often identify $G(\mathcal{P})$ with $\mathcal{P}$ and vice versa.
A distributed algorithm has computed a path decomposition $\mathcal{P}$ if every node knows the paths of $\mathcal{P}$ it belongs to.
Note that it is trivial to compute a $(d(v),1)$-path decomposition in 1 round, because every edge can form a separate path.

Let $\lftwo \cdot \rftwo$ denote the function which rounds down to the previous even integer, that is, $\lftwo x \rftwo = 2 \lfloor x/2 \rfloor$.
The following virtual graph transformation, which we call \emph{edge contraction}, is the core technical construction in this section.

\paragraph{Disjoint Edge Contraction.}
The basic idea behind edge contraction is to turn two incident edges $\{v, u\}$ and $\{v, w\}$ into a single edge $\{u, w\}$ by removing the edges $\{v, u\}$ and $\{v, w\}$ and adding a new edge $\{u, w\}$.
We say that node $v$ contracts when an edge contraction is performed on some pair of edges $\{v, u\}$ and $\{v, w\}$.
When node $v$ performs a contraction of edges $\{v, u\}$ and $\{v, w\}$, its degree $d(v)$ is reduced by two while maintaining the degrees of $u$ and $w$. 
Notice that adjacent nodes can only contract edge-disjoint pairs of edges  in parallel and a contraction may also produce isolated nodes, multi-edges and self-loops. If a self-loop $\{v,v\}$ is selected to be contracted with any other edge $\{v,w\}$ it simply results in a new edge $\{v,w\}$ as if the self-loop was any other edge. Such a contraction still reduces the degree of $v$ by two as the self-loop was considered as both -- an incoming and an outgoing edge of $v$.
See \Cref{fig:contract} for an illustration.

Edge contractions can be used to compute path decompositions, e.g.,  an edge which is created through a contraction of two edges can be seen as a path of length two. If an edge $\{u,v\}$ represents a path from $u$ to $v$ in $G$, e.g.,  when recursively applying edge contractions on the graph $G(\mathcal{P})$ for some given path decomposition $\mathcal{P}$, each contraction merges two paths of $\mathcal{P}$. If each node simply picked arbitrarily some edges to contract, this
might result in long paths or cycles.
The key idea is to use orientations of the edges to find large sets of edges which can be contracted in parallel. If every node only contracts outgoing edges of a given orientation all contractions of all nodes can be performed in parallel.

If we start with a trivial decomposition, i.e., each edge is its own path, and perform $k$ iterations of parallel contraction, where, in each iteration, each node contracts two edges, we obtain a $(d(v) - 2k, 2^k)$-path decomposition.
If we want the degrees $d(v)-2k$ to be constant we have to choose $k$, i.e., the number of iterations, in the order of $\Delta$ which implies exponentially long paths and runtime as the path lengths (might) double with each contraction.

The technical challenge to avoid exponential runtime is to achieve a lot of parallelism while at the same time reducing the degrees quickly. We achieve this with the help of weak orientation algorithms: An outdegree of $f(v)$ at node $v$ allows the node to contract $\lftwo f(v)\rftwo$ edges at the same time and in parallel with all other nodes. If $f(v)$ is a constant fraction of $d(v)$ this implies that $\bigO(\log \Delta)$ iterations are sufficient to reach a small degree.
As the runtime is exponential in the number of iterations and the constant in the $\bigO$-notation might be large, this is still not enough to ensure a runtime which is linear in $\Delta$, up to polylogarithmic terms.
Instead, we begin with the weak orientation algorithm from the previous section and iterate it until a path decomposition with a small (but not optimal!) degree is obtained. Then we use it to construct a better orientation algorithm. Then, we use this better orientation to compute an even better one and so on. Recursing with the correct choice of parameters leads to a runtime which is linear in $\Delta$, up to polylogarithmic terms.
 We take the liberty to use the terms recursion and iteration interchangeably depending on which term is more suitable in the respective context.
Refer to \Cref{fig:contractOne} for an illustration of the edge contraction technique with a given orientation.

\begin{figure}
	\centering
	\begin{tabular}{@{}c@{\hspace{10mm}}|@{\hspace{10mm}}c@{}}
		\includegraphics[scale=0.8]{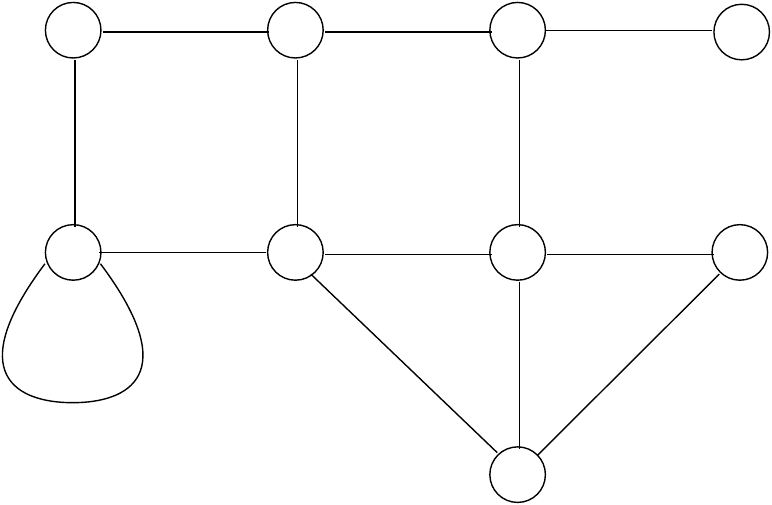} &
		\includegraphics[scale=0.8]{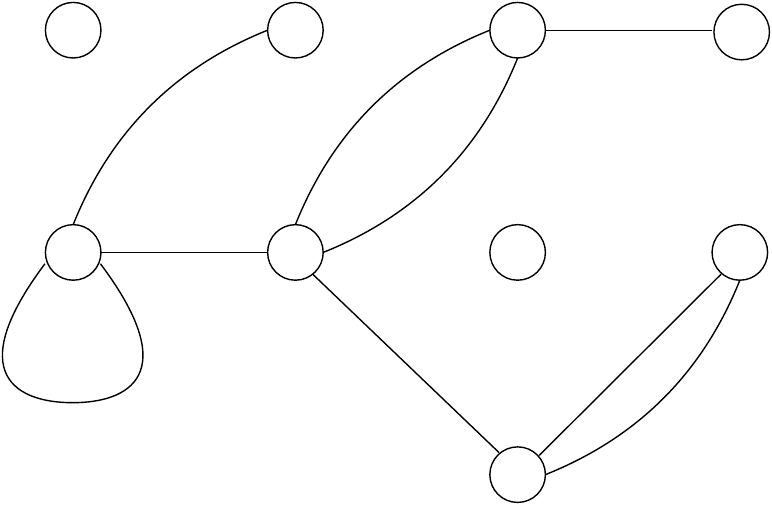} \\[10mm]
		\includegraphics[scale=0.8]{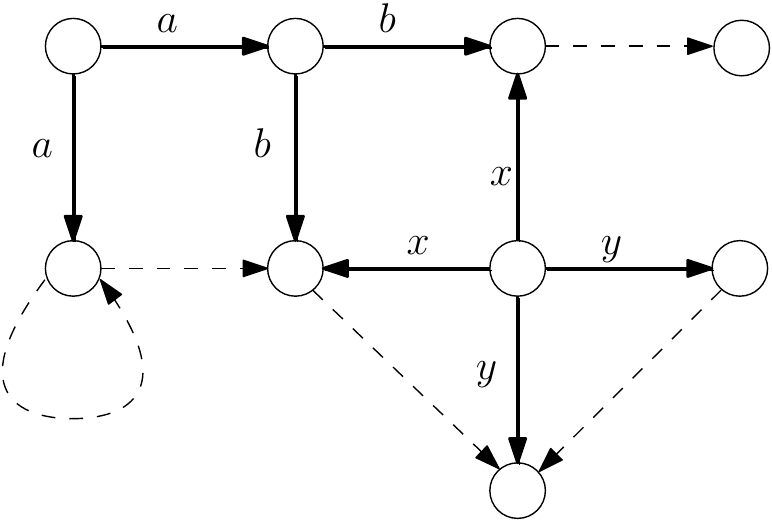} &
		\includegraphics[scale=0.8]{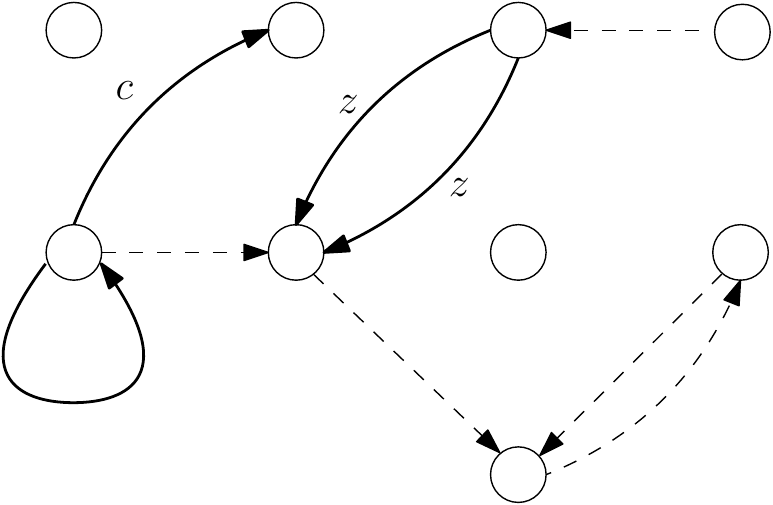} \\[10mm]
		\includegraphics[scale=0.8]{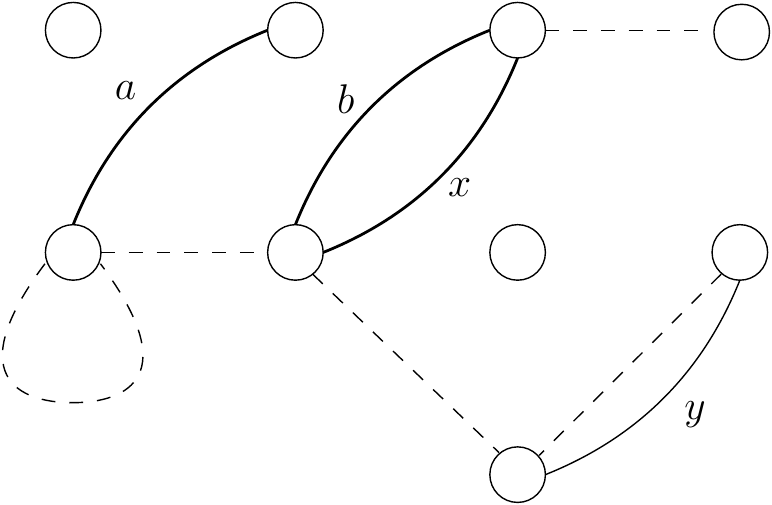} &
		\includegraphics[scale=0.8]{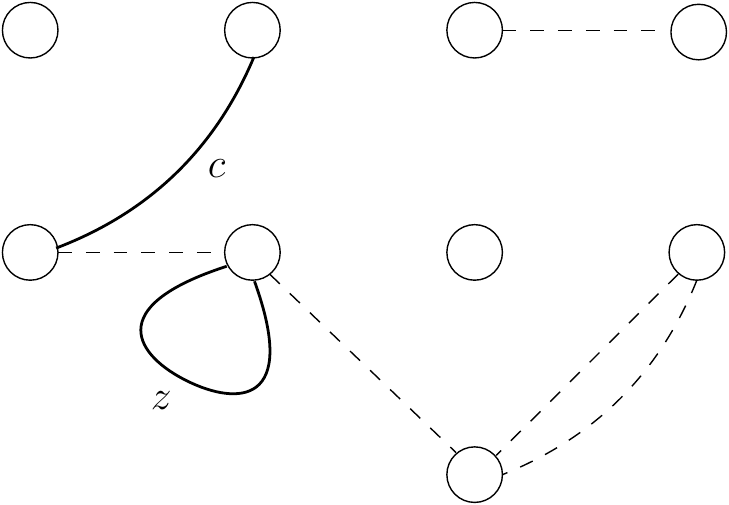}
	\end{tabular}
	\caption{In two sequences of three illustrations this figure depicts two sets of contractions. In each column the first illustration is the situation before the contraction, the second one depicts the orientation and the selected outgoing edges which will be contracted in parallel and the third illustration shows the situation after the contraction where new edges are highlighted.\\
	\hspace*{2.5ex}A contraction may produce isolated nodes, multi-edges and self-loops. If a self-loop $\{v,v\}$ is selected to be contracted with any other edge $\{v,w\}$ it simply results in a new edge $\{v,w\}$ as if the self-loop was any other edge. Such a contraction still reduces the degree of $v$ by two.\\
	\hspace*{2.5ex}Note that we used a graph with small node degrees for illustration purposes. We cannot quickly compute an orientation with large outdegree for nodes with degree	less than five.}
	\label[figure]{fig:contract}
\end{figure}

\begin{figure}
	\centering
	\includegraphics[scale=0.4]{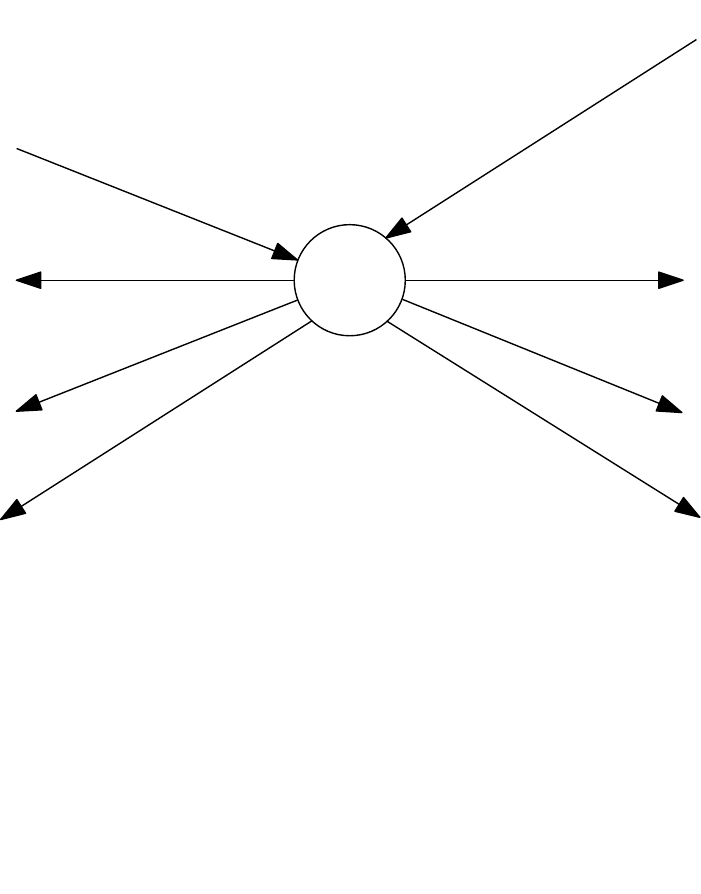}\label[figure]{fig:contractOne1}
	\hspace{\stretch{1}}
	\includegraphics[scale=0.4]{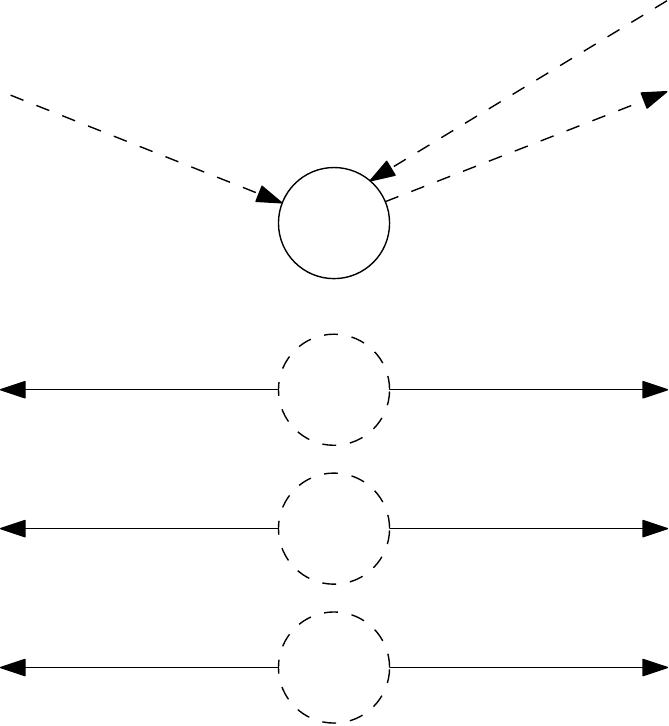}\label[figure]{fig:contractOne2}
	\hspace{\stretch{1}}
	\includegraphics[scale=0.4]{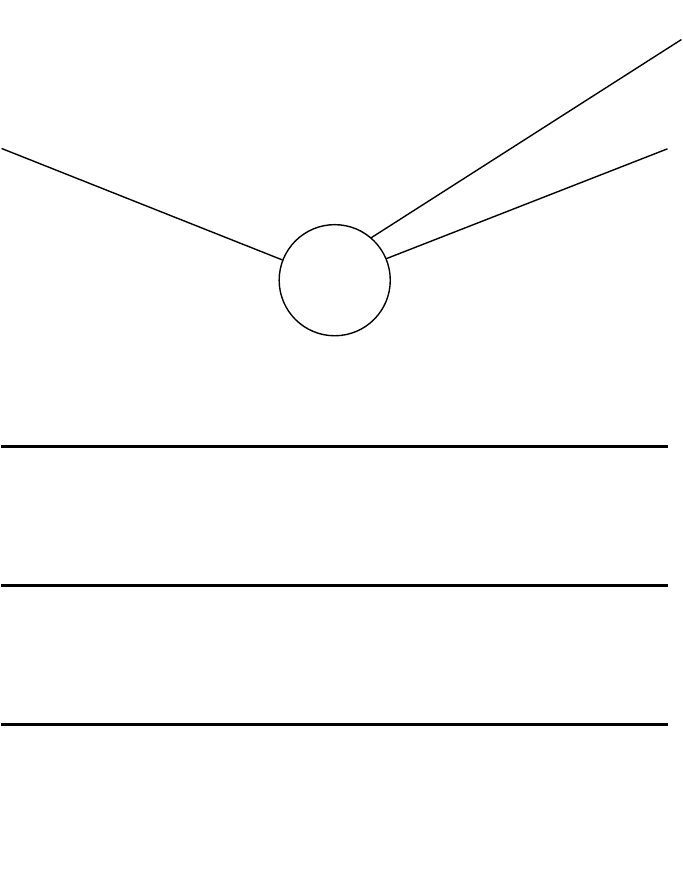}\label[figure]{fig:contractOne35}
	\hspace{\stretch{1}}
	\includegraphics[scale=0.4]{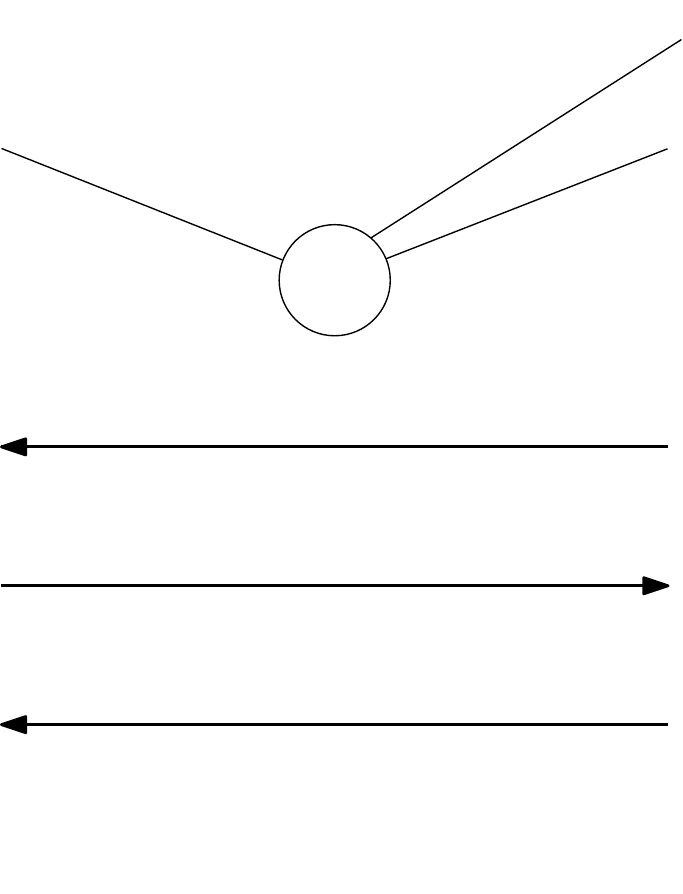}\label[figure]{fig:contractOne3}
	\hspace{\stretch{1}}
	\includegraphics[scale=0.4]{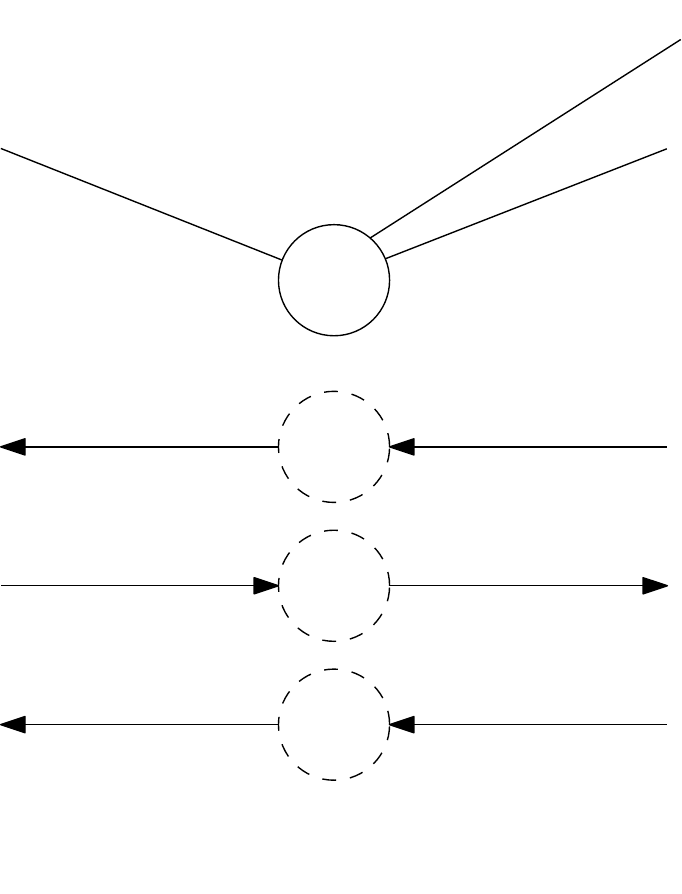}\label[figure]{fig:contractOne4}
	\caption{The first two illustrations show that selecting the outgoing edges for a contraction can be seen as dividing the node into a set of virtual nodes, each incident to two outgoing edges. Then, in the third illustration, the contraction is obtained by removing the virtual nodes but keeping the connection alive. The last two illustrations show how an orientation on contracted edges is used to orient the edges of the original graph such that virtual nodes obtain  an equal split (and such that the original node obtains a good split).}
	\label[figure]{fig:contractOne}
\end{figure}

We will now apply a simple version of our contraction technique to obtain a fast and precise path decomposition algorithm in $\Delta$-regular graphs for $\Delta = \bigO(1)$. The result can also be formulated for non-regular graphs, but here we choose regular graphs to focus on the proof idea which is the key theme throughout most proofs of this section.

\begin{theorem}[$(\Delta-2k,2^k)$-Path Decomposition]
\label[theorem]{thm:decompExpDeltaRuntime}
Let $G=(V,E)$ be a $\Delta$-regular multigraph. For any positive integer $k \leq \Delta/2-2$ there is a deterministic distributed algorithm that computes a $(\Delta-2k,2^k)$-path decomposition in time $\bigO(2^k \log n)$.
\end{theorem}
\begin{proof}
We recursively compute $k$ multigraphs $H_1, \ldots, H_k$ where $H_k$ corresponds to the resulting path decomposition.
To obtain $H_1$, we begin by computing a weak $2$-orientation $\pi$ of $G$ with the algorithm from \Cref{lemma:weakDelta3} (note that by assumption we have $k \ge 1$ and therefore $\Delta \ge 6$).
Then, every node contracts a pair of outgoing incident edges.
Notice that contractions of adjacent nodes are always disjoint.
The degree of of every node is reduced to $\Delta - 2$ and each edge in the resulting multigraph $H_1$ consists of a path in $G$ of length at most two.

Applying this method recursively with recursion depth $k$ yields multigraphs $H_1,\ldots,H_k$ where the maximum degree of $H_i$ is $\Delta-2i$ and each edge in $H_i$ corresponds to a path in $G$ of length at most $2^i$. Thus, $H_k$ corresponds to a $(\Delta-2k,2^k)$-path decomposition. Note that there is one execution of \Cref{lemma:weakDelta3} in each recursion level and it provides a weak $2$-orientation of the respective graph because the degree of each node is at least six due to $i\leq k\leq \Delta/2-2$.

One communication round in recursion level $i$ can be simulated in $2^i$ rounds in the original graph.
Thus, the runtime is dominated by the application of \Cref{lemma:weakDelta3} in recursion level $k$ which yields a time complexity of  $\bigO(2^k \log n)$.
\end{proof}

Next, we show how to turn a $(\delta,\lambda)$-path decomposition efficiently into a strong orientation.
The strong orientation obtained this way has $\delta(v)$ as an upper bound on the discrepancy between in- and outdegree of node $v$.

\begin{lemma}
\label[lemma]{lemma:fromPathDecompToStrongOrient}
Let $G=(V,E)$ be a multigraph with a given $(\delta,\lambda)$-path decomposition $\mathcal{P}$. There is a deterministic algorithm that computes a strong $\frac{1}{2}(d(v)-\delta(v))$-orientation of $G$ in $\bigO(\lambda)$ rounds.
\end{lemma}
\begin{proof}
Let $H=G(\mathcal{P})$ be the virtual graph that corresponds to $\mathcal{P}$ and let $\pi_H$ be an arbitrary orientation of the edges of $H$.
Let $(u, v)$ be an edge of $H$ oriented according to $\pi_H$ and let $P = u_1, \ldots, u_k$, where $u_1 = u$ and $u_k = v$, be the path in the original graph $G$ that corresponds to edge $(u, v)$ in $H$.
Now, we orient the path $P$ in a consistent way according to the orientation of $(u, v)$, i.e., edge $\{u_i, u_{i + 1}\}$ is directed from $u_i$ to $u_{i + 1}$ for all $1 \leq i \leq k$.
Since every edge in $G$ belongs to exactly one path in the decomposition, performing this operation for every edge in $H$ provides a unique orientation for every edge in $G$.
Let us denote the orientation obtained this way by $\pi_G$.

Consider some node $v$ and observe that orienting any path that contains $v$ but where $v$ is not either the start or the endpoint adds exactly one incoming edge and one outgoing edge for~$v$.
Therefore, the discrepancy of the indegrees and outdegrees of $v$ in $\pi_G$ is bounded from above by the discrepancy in $\pi_H$, which is at most $\delta(v)$ by the definition of a $(\delta,\lambda)$-path decomposition. It follows that $\pi_G$ is a strong $\frac{1}{2}(d(v)-\delta(v))$-orientation.

Finally, since the length of any path in $\mathcal{P}$ is bounded above by $\lambda$, consistently orienting the paths takes $\lambda$ communication rounds finishing the proof.
\end{proof}

In the following, we formally use weak orientations to compute a path decomposition. This lemma will later be iterated in \Cref{cor:iterative}.
\begin{lemma}
\label[lemma]{lemma:newDecomp}
	 Assume that there exists a deterministic distributed algorithm that finds a weak $\bigl(\bigl(\frac{1}{2}-\eps\bigr)d(v)- 2\bigr)$-orientation in time $T(n,\Delta)$.

Then, there is a deterministic distributed algorithm that finds a $\bigl(\bigl(\frac{1}{2}+\eps\bigr)d(v)+4,\, 2\bigr)$-path decomposition in time $\bigO(T(n,\Delta))$.
\end{lemma}
\begin{proof}
	Let $G$ be a multigraph with a weak $\bigl(\bigl(\frac{1}{2}-\eps\bigr)d(v)-2\bigr)$-orientation given by the algorithm promised in the lemma statement.
	Now every node $v$ arbitrarily divides the outgoing edges into pairs and contracts these pairs yielding a multigraph with degree at most 
	\[\textstyle d(v) - \lftwo\bigl(\frac{1}{2}-\eps\bigr)d(v)\rftwo + 2 \leq \bigl(\frac{1}{2} + \eps\bigr)d(v) + 4.\]
	Observing that all of the chosen edge pairs are disjoint yields that the constructed multigraph is a $\bigl(\bigl(\frac{1}{2}+\eps\bigr)d(v)+4,\, 2\bigr)$-path decomposition. 
	The contraction operation requires one round of communication.
\end{proof}

In the following lemma we iterate \Cref{lemma:newDecomp} to obtain an even better path decomposition. Furthermore, more care is required in the details to avoid rounding errors and to obtain the correct result when the degrees get small. \Cref{cor:iterative} will be applied many times in proceeding subsections.
\begin{lemma}
\label[lemma]{cor:iterative}
Let $0<\eps\leq 1/6$.
Assume that $T(n,\Delta)\geq \log n$ is the running time of an algorithm $\mathcal{A}$ that finds a weak $\bigl((1/2-\eps)d(v)-2\bigr)$-orientation.
Then for any positive integer~$i$, there is a deterministic distributed algorithm $\mathcal{B}$ that finds a $\bigl((1/2+\eps)^id(v)+4,\, 2^{i+5} \bigr)$-path decomposition $\mathcal{P}$ in time $\bigO(2^i\cdot T(n,\Delta))$.
\end{lemma}
\begin{proof}
Let $i$ be a positive integer. We define algorithm $\mathcal{B}$ such that it uses algorithm $\mathcal{A}$ to recursively compute graphs $H_0,H_1,\dotsc,H_i, H_{i+1}, \dotsc, H_{i+5}$ and path decompositions $\mathcal{P}_1,\mathcal{P}_2,\dotsc,\allowbreak\mathcal{P}_{i},\allowbreak \mathcal{P}_{i+1},\allowbreak\dotsc,$ $ \mathcal{P}_{i+5}$. Let $G=(V,E)$ be a multigraph. For $j=0,\ldots, i-1$ we set $H_0=G$ and $H_{j+1}=H_j(\mathcal{P}_{j+1})$, where $\mathcal{P}_{j+1}$ is the path decomposition which is returned by applying \Cref{lemma:newDecomp} with algorithm $\mathcal{A}$ on $H_j$. This guarantees that path decomposition $\mathcal{P}_i$ has maximum degree $(\frac{1}{2}+\eps)^id(v)+12$. The remaining five graph decompositions are computed afterwards (see the end of this proof) and reduce the additive $12$ to an additive $4$.

\subparagraph{Properties of \boldmath$\mathcal{P}_{1}, \ldots, \mathcal{P}_{i}$.}
We first show that for $j=1,\ldots,i$ the path decomposition $\mathcal{P}_j$ is a $(z_j(v), 2^j)$-path decomposition with
\[z_j(v)=\bigl(\tfrac{1}{2}+\eps\bigr)^jd(v)+4\sum_{k=0}^{j-1}\bigl(\tfrac{1}{2}+\eps\bigr)^k.\]
With every application of \Cref{lemma:newDecomp} the length of the paths at most doubles in length which implies that the path length of $\mathcal{P}_j$ is upper bounded by $2^j$.
We now prove by induction that the variables $z_j(v)$, $j=1,\ldots i$ behave as claimed:
\begin{itemize}
	\item \emph{Base case:} $z_1(v)=\big(\frac{1}{2}+\eps\big)d(v)+4$ follows from the invocation of \Cref{lemma:newDecomp} with $\mathcal{A}$ on $H_0=G$.
	\item \emph{Inductive step:} Using the properties of \Cref{lemma:newDecomp} we obtain
\begin{align*}
z_{j+1}(v) & =\bigl(\tfrac{1}{2}+\eps\bigr)z_j(v)+4
\leq \bigl(\tfrac{1}{2}+\eps\bigr)\biggl(\bigl(\tfrac{1}{2}+\eps\bigr)^jd(v)+4\sum_{k=0}^{j-1}\bigl(\tfrac{1}{2}+\eps\bigr)^k\biggr)+4\\
& =\bigl(\tfrac{1}{2}+\eps\bigr)^{j+1}d(v)+4\sum_{k=0}^{j}\bigl(\tfrac{1}{2}+\eps\bigr)^k\text{.}
\end{align*}
\end{itemize}
Using the geometric series to bound the last sum and then $\eps\leq 1/6$ we obtain that
\[z_i(v)\leq \bigl(\tfrac{1}{2}+\eps\bigr)^id(v)+12.\]

\subparagraph{Reducing the Additive Term.}
Now, we compute the five further path decompositions $\mathcal{P}_{i+1}, \ldots, \mathcal{P}_{i+5}$ to reduce the additive term in the degrees of the path decomposition from $12$ to $4$; in each path decomposition this additive term is reduced by two for certain nodes. In each of the first four path decompositions nodes with degree at least six in the current path decomposition reduce the additive term by at least two: we compute a weak $\lfloor d(v)/3\rfloor$-orientation (using \Cref{lemma:weakDelta3}) and then every node with degree at least six contracts two outgoing edges.
In the last path decomposition we compute an orientation in which every node with degree at least five in the current path decomposition has two outgoing edges (using \Cref{lemma:firstoutdeg-2}) and then each of them contracts two incident edges. Thus in the last path decomposition the additive term of nodes with degree five is reduced by two.

To formally prove that we obtain the desired path decomposition let $x_{i+j}(v)$ be the actual degree of node $v$ in $G(\mathcal{P}_{i+j})$ for $j=0,\ldots, 5$. First note that the degree of a node never increases due to an edge contraction, not even due to an edge contraction which is performed by another node.

\subparagraph{Constructing \boldmath$\mathcal{P}_{i+1}, \ldots, \mathcal{P}_{i+4}$.}
To determine path decomposition $\mathcal{P}_{i+j+1}$ for $j=0,\ldots,3$, we compute an orientation of $G(\mathcal{P}_{i+j})$ in which every node $v$ with $x_{i+j}(v) \geq 6$ has outdegree at least two (one can use the algorithm described in \Cref{lemma:weakDelta3}). Then $\mathcal{P}_{i+j+1}$ is obtained if every node with $x_{i+j}(v)\geq 6$ contracts two of its incident outgoing edges.  So, whenever $x_{i+j}(v)\geq 6$ we obtain that $x_{i+j+1}(v)= x_{i+j}(v)-2$, that is $x_{i+j+1}\leq z_i(v)-2(j+1)$. If $x_{i+j}(v)\geq 6$ for all $j=0,\ldots,3$ we have
\[x_{i+5}(v)\leq x_{i+4}(v)\leq (1/2+\eps)^id(v)+4.\]
Otherwise, for some $j=0,\ldots,3$, we have $x_{i+j}(v)\leq 5$, that is, $x_{i+4}(v)\leq 4$ or $x_{i+4}(v)=5$. If $x_{i+4}(v)\leq 4$ we have
\[x_{i+5}(v)\leq x_{i+4}(v)\leq 4\leq (1/2+\eps)^id(v)+4.\]

\subparagraph{Constructing \boldmath$\mathcal{P}_{i+5}$.}
For nodes with $x_{i+4}(v)=5$ we compute one more path decomposition.
We use \Cref{lemma:firstoutdeg-2} to compute an orientation of $G(\mathcal{P}_4)$ in which each node with degree at least five has two outgoing edges; then each node with at least two outgoing edges contracts one pair of its incident outgoing edges.
Thus the degree of nodes with degree five reduces by two and we obtain that the path decomposition  $\mathcal{P}_{i+5}$  is a  $\bigl((\frac{1}{2}+\eps)^id(v)+4,\, 2^{i+5}\bigr)$-path decomposition.

\subparagraph{Running Time.}
The time complexity to invoke algorithm $\mathcal{A}$ or the algorithms from \Cref{lemma:weakDelta3} or \Cref{lemma:firstoutdeg-2}  on graph $H_j$ is $\bigO(2^jT(n,\Delta))$ because the longest path in $H_j$ has length $2^j$ and $T(n,\Delta)\geq \log n$. Thus, the total runtime is
\[\bigO\biggl(\,\sum_{j=0}^{i+5}2^jT(n,\Delta)\biggr)=\bigO\bigl(2^{i}T(n,\Delta)\bigr).\qedhere\]
\end{proof}
The reduction of the additive term in the proof of \Cref{cor:iterative} is most likely not helpful for edge coloring applications as constant degree graphs can be colored quickly anyways. However, for theoretical reasons it is interesting to see how close we can get to optimal splits with regard to the discrepancy. The splits that we obtain for directed splitting are optimal; the undirected splitting result leaves a bit of space for improvement. 

\subsection{Amplifying Weak Orientation Algorithms}
Now, we use \Cref{cor:iterative} to iterate a given weak orientation algorithm $\mathcal{A}$ to obtain a new weak orientation algorithm $\mathcal{B}$. The goal is that $\mathcal{B}$ has an outdegree guarantee which is much closer to $(1/2) d(v)$ than the guarantee provided by algorithm $\mathcal{A}$.

\begin{lemma}\label[lemma]{lemma:weakTransformation}
Let $0<\eps_2<\eps_1\leq \frac{1}{6}$.
Assume that there is a deterministic algorithm  $\mathcal{A}$ which computes a  weak $\left(\left( \frac{1}{2}-\eps_1 \right) d(v)-2\right)$-orientation and runs in time $T(n,\Delta)$.
Then there is a deterministic weak $\left(\left( \frac{1}{2}-\eps_2 \right)d(v)-2 \right)$-orientation  algorithm $\mathcal{B}$ with running time
\begin{align}
\label[equation]{eqn:runtimeweakTransformation}
	\bigO\Bigl(\eps_2^{\log_2^{-1}( \frac{1}{2}+\eps_1 )}\cdot T(n,\Delta)\Bigr)=\bigO\Bigl(\eps_2^{-(1+24\eps_1)}\cdot T(n,\Delta)\Bigr).
\end{align}
\end{lemma}

Let $\alpha = \frac{1}{2}-\eps_1$ and $\beta = \frac{1}{2} + \eps_1$. The roadmap for the proof of Lemma \ref{lemma:weakTransformation} is as follows:
\begin{enumerate}[label=(\arabic*)]
	\item Execute $i$ iterations of a weak $(\alpha d(v) - 2)$-orientation algorithm, for an $i$ that will be chosen later, and after each iteration, perform disjoint edge contractions. Thus, we obtain a $\bigl(\beta^{i} d(v) + 4,\, 2^{i + 5} \bigr)$-path decomposition using \Cref{cor:iterative}.
	\item Apply \Cref{lemma:fromPathDecompToStrongOrient} to obtain a strong (and thus also a weak) $\bigl( \frac{1}{2} ( 1 -  \beta^i ) d(v) - 2\bigr)$-orientation.
	\item By setting $i = \log(\eps_2) / \log(\beta)$ we get that $\beta^{i} = \eps_2$ and the running time of steps 1--2 is
	\[
		\bigO(2^i T(n,\Delta)) = \bigO\bigl(\eps_2^{\log_2^{-1}{\beta }}\cdot T(n,\Delta)\bigr) = \bigO\bigl(\eps_2^{-(1 + 24\eps_1)}\cdot T(n,\Delta)\bigr),
	\]
	where $T(n, \Delta)$ is the runtime of the weak $(\alpha d(v) - 2)$-orientation algorithm. The last equality holds because with \Cref{lemma:taylor}, we obtain that $-\log_2^{-1} \beta \leq 1+24\eps_1$ when $\eps_1 \leq 1/6$.
\end{enumerate}
With \Cref{lemma:weakTransformation} at hand we can amplify the quality of splitting algorithms and obtain the following theorem.
\begin{theorem}
\label[theorem]{thm:orientationAlgorithms}
	Let $\discr$ be a positive integer. There exist the following deterministic weak orientation algorithms.
\begin{enumerate}[label=(\alph*)]
\item $\mathcal{A}$: weak $\left( \left( \frac{1}{2}-1/\log\log \frac{\Delta}{\discr} \right)d(v)- 2 \right)$-orientation in time $\bigO\bigl( \bigl(\log\log \frac{\Delta}{\discr} \bigr)^{1.71} \cdot \log n \bigr)$.
\item $\mathcal{B}$: weak $\left( \left( \frac{1}{2}-1/\log \frac{\Delta}{\discr} \right)d(v)- 2 \right)$-orientation in time $\bigO \bigl( \log \frac{\Delta}{\discr} \cdot \bigl(\log\log{\frac{\Delta}{\discr}}\bigr)^{1.71} \cdot \log n \bigr)$.
\item $\mathcal{C}$: weak $\left( \left( \frac{1}{2}-\frac{\discr}{\Delta} \right)d(v)- 2 \right)$-orientation in time $\bigO \bigl(\frac{\Delta}{\discr} \cdot \log \frac{\Delta}{\discr} \cdot \bigl( \log\log{\frac{\Delta}{\discr}} \bigr)^{1.71} \cdot \log n \bigr)$.
\end{enumerate}
\end{theorem}
In the proof of \Cref{thm:orientationAlgorithms} we perform the following steps:
\begin{enumerate}[resume*]
	\item Use \Cref{lemma:weakTransformation} with $\eps_1 = 1/6$ and $\eps_2 = 1/\log \log \Delta$ to obtain an algorithm which computes a weak $\bigl( \bigl( \frac{1}{2} - 1/\log \log \Delta \bigr) d(v) - 2 \bigr)$-orientation and runs in time $\bigO((\log \log \Delta)^{1.71} \cdot \log n)$. In this step, we plug in $\eps_1 = 1/6$ to obtain the exponent
	\[-\log_2^{-1} \beta = -\log_2^{-1} \bigl(\tfrac{1}{2} + \tfrac{1}{6}\bigr) < 1.71.\]
	\item Using the construction twice more, once with $\eps_1=1/\log\log\Delta$ and $\eps_2=1/\log\Delta$ and once with  $\eps_1 = 1/\log \Delta$ and $\eps_2 = 1 / \Delta$, yields a  
	\[\text{weak }\left( \left( \frac{1}{2} - \frac{1}{\Delta} \right) d(v) - 2 \right)\text{-orientation algorithm}\] that runs in time $\bigO(\Delta \cdot \log \Delta \cdot (\log \log \Delta)^{1.71} \cdot \log n)$.
\end{enumerate}

Before we continue with the formal proofs of \Cref{lemma:weakTransformation,thm:orientationAlgorithms} we prove the following technical result that we use to simplify running times; it is proved with a Taylor expansion.
\begin{lemma}
	\label[lemma]{lemma:taylor}
	Let $0<\eps \leq 1/6$. Then, $-\log^{-1}_2{(\frac{1}{2} + \eps)} \leq 1 + 24\eps$.
\end{lemma}
\begin{proof}
Let $z = 2\eps/\bigl(\frac{1}{2} + \eps\bigr) \leq 4\eps$.
Notice that $2 - z = \bigl(\frac{1}{2} + \eps\bigr)^{-1}$.
By writing $\log_2^{-1}(2 - z)$ using Taylor series at $2$, we get that
\begin{align*}
-\log_2^{-1}\bigl(\tfrac{1}{2} + \eps \bigr) & = \log_2^{-1}(2-z) = \ln(2) \ln^{-1}(2 - z) \\
& = \ln(2)\biggl(\ln(2)-\sum_{k=1}^{\infty}\frac{1}{k\cdot2^k}z^k\biggr)^{-1}\\
& \leq \biggl(1-\ln^{-1} 2 \sum_{k=1}^{\infty}\frac{1}{2^k }z^k\biggr)^{-1}\\
& \stackrel{|z|<1}{\leq} \biggl(1-z\ln^{-1} 2\sum_{k=1}^{\infty}\frac{1}{2^k }\biggr)^{-1}\\
& = (1-z\ln^{-1}2)^{-1} \leq 1+z\cdot\frac{\ln^{-1}2}{1-z\ln^{-1}2}\\
& \stackrel{\eps \leq 1/6}{\leq} 1+6z\stackrel{z \leq 4\eps} \leq 1+24\eps.
\qedhere
\end{align*}
\end{proof}

In the following proof we perform steps 1--3 of the aforementioned agenda.

\begin{proof}[Proof of \Cref{lemma:weakTransformation}]
Let $i = \log_2(\eps_2) / \log_2{(1/2+\eps_1)}$ which is bounded by $(1 + 24\eps_1) \log_2(1/\eps_2)$ due to \Cref{lemma:taylor}; thus it is sufficient to show the left hand side of $(\ref{eqn:runtimeweakTransformation})$.
By applying \Cref{cor:iterative} with parameter $i$ and algorithm $\mathcal{A}$, we get a distributed algorithm that finds a 
\[\left(\left(1/2 + \eps_1\right)^r d(v) + 4,\, 2^{i+5} \right)\text{-path decomposition}\] in time
\[
 	\bigO\bigl(2^i \cdot T(n, \Delta)\bigr) = \bigO\Bigl(\eps_2^{\log_2^{-1}( \frac{1}{2}+\eps_1 )}\cdot T(n,\Delta)\Bigr).
\]
The degree of node $v$ in the path decomposition is upper bounded by
\[\bigl(\tfrac{1}{2} + \eps_1\bigr)^i d(v)+4=\eps_2d(v)+4.\]
Now \Cref{lemma:fromPathDecompToStrongOrient} yields a weak $\bigl( \frac{1}{2}\bigl(1 - \eps_2\bigr) d(v) - 2\bigr)$-orientation algorithm with the same running time; in particular, this is a weak $\bigl( \bigl(\frac{1}{2} - \eps_2\bigr) d(v) - 2\bigr)$-orientation algorithm.
\end{proof}

We close the section by performing steps 4--5 of the agenda. Note that the theorem is more general than what was outlined in the agenda as it contains an additional parameter $\discr$ which can be used to tune the running time at the cost of the quality of the weak orientation algorithm.

\begin{proof}[Proof of \Cref{thm:orientationAlgorithms}]Each statement is proven by applying \Cref{lemma:weakTransformation} with different values for $\eps_1$ and $\eps_2$.
\begin{enumerate}[label=(\alph*)]
\item
We obtain the algorithm $\mathcal{A}$ by applying \Cref{lemma:weakTransformation} with the weak $\lfloor \Delta/3\rfloor$-orientation algorithm from \Cref{lemma:weakDelta3}, that is with $\eps_1=1/6$, and with $\eps_2 = 1/\log\log (\Delta / \delta)$.

\item Algorithm $\mathcal{B}$ is obtained by applying \Cref{lemma:weakTransformation} with the algorithm $\mathcal{A}$ from part (a) as input (i.e., $\eps_1 = 1/\log\log (\Delta / \delta)$) and with $\eps_2 = 1/\log (\Delta / \delta)$.
\item Algorithm $\mathcal{C}$ is obtained by applying \Cref{lemma:weakTransformation} with the algorithm $\mathcal{B}$ from part (b) as input (i.e., $\eps_1=1/\log (\Delta / \delta)$) and with $\eps_2=1/(\Delta / \delta) = \delta / \Delta$.
\qedhere
\end{enumerate}
\end{proof}

\subsection{Short and Low Degree Path Compositions Fast}

Our higher level goal is to compute a path decomposition where the degree is as small as possible to obtain a directed split with the discrepancy as small as possible (with methods similar to \Cref{lemma:fromPathDecompToStrongOrient}, also see the proof of \Cref{thm:mainSplitting}).
As we will show in the next theorem, with the methods introduced in this section and the appropriate choice of parameters, we can push the maximum degree of the path decomposition down to $\eps d(v) + 4$ for any $\eps>0$. This is the true limit of this approach because we cannot compute weak $2$-orientations of $4$-regular graphs in sublinear time (see \Cref{thm:weak2orientation-lb}).
\newcommand{\Z}{\alpha}

\begin{theorem}
\label[theorem]{lemma:mainPathDecomposition}
Let $G=(V,E)$ be a multigraph with maximum degree $\Delta$. For any $\eps>0$  there is a deterministic distributed algorithm which computes a $(\delta(v),\bigO(1/\eps))$-path decomposition in time
$\bigO\bigl(\Z \cdot \log \Z \cdot (\log\log\Z)^{1.71} \cdot \log n\bigr)$, where $\Z=2/\eps$ and $\delta(v)=\eps d(v)+3$ if $\eps d(v)\geq1$ and $\delta(v)=4$ otherwise.

\end{theorem}
\begin{proof} Apply \Cref{cor:iterative} with algorithm $\mathcal{B}$ from \Cref{thm:orientationAlgorithms}, $\delta=\Delta/\Z$, and
\[i=\frac{\log{\Z^{-1}}}{\log(1/2+1/\log (\Z))}.\]
This implies a path decomposition with degrees  $\lfloor\Z^{-1} d(v)+4\rfloor=\lfloor \eps d(v)/2+4\rfloor$. If $\eps d(v)\geq 1$ this is smaller than $\eps d(v)+3$. If $\eps d(v)<1$ this is at most $4$.
The length of the longest path is upper bounded by $\bigO(2^i)=\bigO\bigl(\Z^{1+24/\log{\Z}}\bigr)=\bigO(\alpha)$ where we used \Cref{lemma:taylor}.
The runtime is bounded by
\[\bigO\bigl(2^i \cdot T_{\mathcal{B}}(n,\Delta)\bigr) = \bigO\bigl(\Z \cdot \log \Z \cdot \left(\log\log \Z\right)^{1.71} \cdot \log n \bigr),\]
 where $T_{\mathcal{B}}(n,\Delta)$ is the running time of algorithm $\mathcal{B}$.
\end{proof}

Choosing $\eps=1/(2\Delta)$ in \Cref{lemma:mainPathDecomposition} yields the following corollary.
\begin{corollary}[Constant Degree Path Decomposition]
There is a deterministic algorithm which computes a
$(4,\bigO(\Delta))$-path decomposition in time $\bigO\bigl(\Delta\cdot \log\Delta\cdot (\log\log\Delta)^{1.71}\cdot \log n\bigr)$.
\end{corollary}

\begin{remark}
For any positive integer $k$ smaller than $\logStar(\alpha)\pm \bigO(1)$ one can improve the runtime of \Cref{lemma:mainPathDecomposition}  to
$\bigO\bigl(\alpha\cdot (\log^{(k)}\alpha)^{0.71}\cdot\log n\cdot \Pi_{j=1}^k\log^{(j)} \alpha\bigr)$, where $\log^{(j)}(\cdot)$ denotes the $j$ times iterated logarithm, $\alpha=2/\epsilon$ and the constant in the $\bigO$-notation grows exponentially in $k$. This essentially follows from a version of \Cref{thm:orientationAlgorithms} that turns a weak $\bigl((1/2-1/\log^{(k)}\alpha)d(v)- 2\bigr)$-orientation algorithm into a weak $\bigl((1/2-1/\log \alpha)d(v)- 2\bigr)$-orientation algorithm in $k-1$ iterations.
\end{remark}

%!TEX root = main.tex

\section[\texorpdfstring{Degree $3$: Sinkless and Sourceless Orientations}{Degree 3: Sinkless and Sourceless Orientations}]{\texorpdfstring{Degree \boldmath$3$: Sinkless and Sourceless Orientations}{Degree 3: Sinkless and Sourceless Orientations}}\label[section]{sec:basecase}
The results of this section are used in the proof of \Cref{thm:mainSplitting} and in \Cref{sec:outdegtwo}. 

First note that an arbitrary consistent orientation of the paths in the best path decomposition of \Cref{sec:shortPathsDecompositions} would result in a  splitting in which each node $v$ has discrepancy at most $\eps\cdot d(v) + 4$. In the case of directed splitting we slightly tune this in the proof of \Cref{thm:mainSplitting} by consistently orienting the paths in such a way that each node has at least one outgoing and one incoming path. As the graph corresponding to the path decomposition is a low degree graph this is the same as finding sinkless and sourceless orientations in low-degree graphs; in this section we show how to compute these. Thereby the most challenging case is to make sure that also nodes of degree three will have one outgoing and one incoming edge. 

The main results of this section are \Cref{lemma:sinkless-sourceless} and the immediate \Cref{corollary:sinkless-sourceless}.
To prove the lemma we will first concentrate on high-girth graphs; then, in \Cref{sec:shortCycles}, we show how to handle short cycles and complete the proof of the following lemma.
\begin{lemma}[Sinkless and Sourceless Orientation]\label[lemma]{lemma:sinkless-sourceless}
The following problem can be solved in time $O(\log n)$ with deterministic algorithms and $O(\log \log n)$ with randomized algorithms:
given a $3$-regular multigraph, find a sinkless and sourceless orientation.
\end{lemma}
With a simple reduction (similar to \Cref{lemma:weakmulti}), we can generalise these results to non-regular graphs as well:
\begin{corollary}[Sinkless and Sourceless Orientation]\label[corollary]{corollary:sinkless-sourceless}
The following problem can be solved in time $O(\log n)$ with deterministic algorithms and $O(\log \log n)$ with randomized algorithms:
given any multigraph, find an orientation such that all nodes of degree at least $3$ have outdegree and indegree at least $1$.
\end{corollary}

\begin{proof}
Let $G$ be any multigraph. First, we split any node of degree $k+3$ into $k$ nodes of degree $1$ and one node of degree $3$. We also split each node of degree $k < 3$ into $k$ nodes of degree $1$. Now we are left with a graph $G'$ in which each node has degree $1$ (these are \emph{leaf nodes}) or $3$ (these are \emph{internal nodes}). Finally, we augment each leaf node with a gadget in order to obtain a $3$-regular graph $G''$ (\Cref{fig:basecase-regularity}).

\begin{figure}
\centering
\includegraphics[scale=0.5]{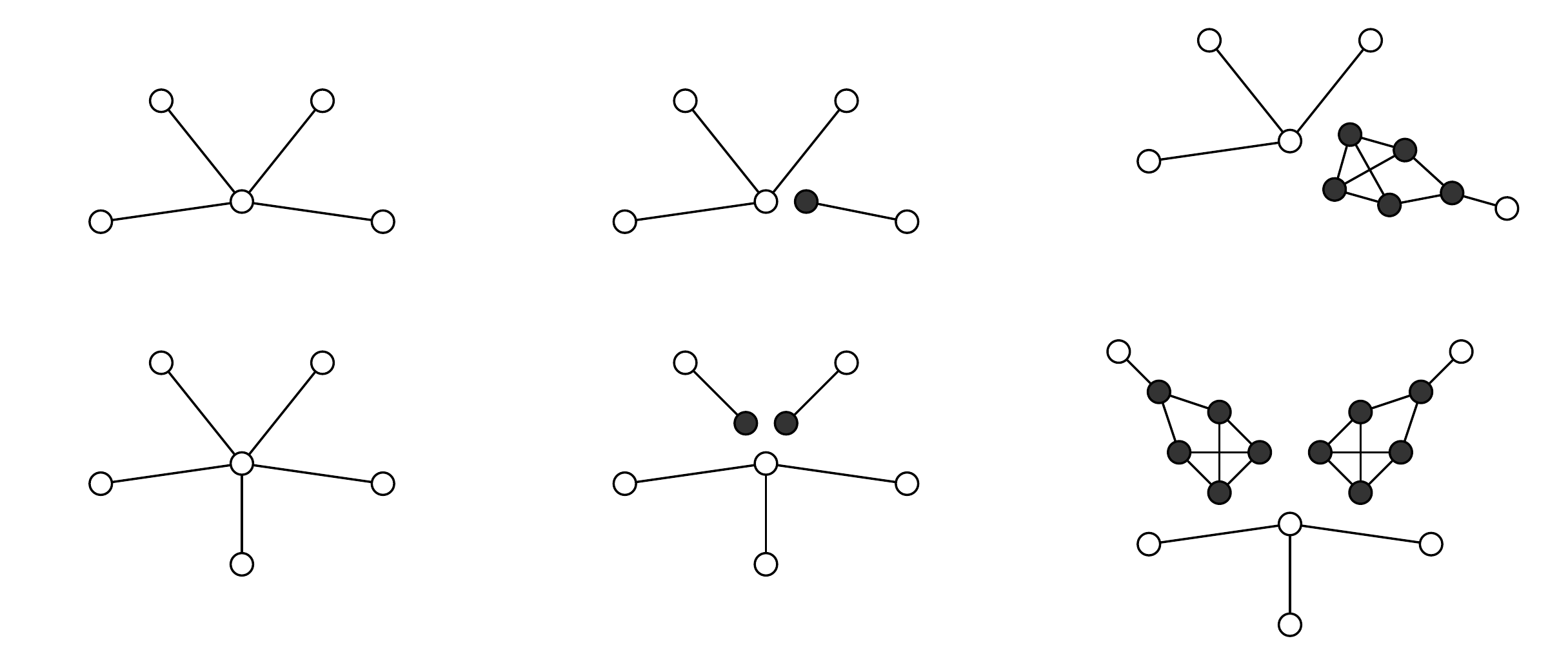}
\caption{Given a general graph $G$, we first split the nodes to obtain a graph $G'$ with degrees $1$ and $3$, and then add gadgets to obtain a $3$-regular graph $G''$.}\label[figure]{fig:basecase-regularity}
\end{figure}

We then invoke \Cref{lemma:sinkless-sourceless} to find a sinkless and sourceless orientation in $G''$. We delete the gadgets to get back to graph $G'$; now each internal node has outdegree at least $1$. Finally, we revert the splitting to get back to graph $G$; now for each node of degree at least $3$, there is one internal node that contributes at least one outgoing edge. Furthermore, computation in $G''$ can be simulated in $G$ with only constant-factor overhead.
\end{proof}

In the proof of \Cref{lemma:sinkless-sourceless}, we will use the following observations:
\begin{itemize}
    \item It is easier to find a sinkless and sourceless orientation if we only care about nodes of degree at least $6$.
    \item In high-girth graphs, we can make the degrees larger with the help of contractions, and then it is sufficient to find orientations that make nodes of degree at least $6$ happy. (These contractions are different from the contractions in \Cref{sec:shortPathsDecompositions} and are like the contractions that are known from building minors)
    \item In low-girth graphs, we can exploit short cycles to find orientations, eliminate them, and then it is sufficient to find orientations in high-girth graphs.
\end{itemize}

\subsection{Degree 6: Sinkless and Sourceless Orientation}

Let us now start with simple observations related to the case of nodes of degree at least $6$. For brevity, let us write $\Tso$ for the time complexity of sinkless orientations in the model that we study: $\Tso = O(\log n)$ for deterministic algorithms and $\Tso = O(\log \log n)$ for randomized algorithms. We start with the following simple lemma (cf.\ \Cref{lemma:weakDelta3}).

\begin{lemma}[Degree 6, Outdegree 2]\label[lemma]{lemma:deg6-outdeg2}
The following problem can be solved in time $O(\Tso)$:
given a graph, find an orientation such that all nodes of degree at least $6$ have outdegree at least $2$.
\end{lemma}
\begin{proof}
Split all nodes of degree at least $6$ into two nodes of degree at least $3$. Apply \Cref{lemma:weakmulti} to find an orientation in which all nodes of degree at least $3$ have outdegree at least $1$. Merge the nodes back.
\end{proof}
A useful interpretation of the above lemma is that each node with degree at least 6 \emph{owns} at least two of its incident edges, i.e., the outgoing edges. Now each node can freely re-orient the edges that it owns whichever way it wants. In particular, each node of degree at least $6$ can make sure that there is at least one outgoing edge and at least one incoming edge:

\begin{corollary}[Degree 6, Sinkless and Sourceless]\label[corollary]{corollary:deg6-sinkless-sourceless}
The following problem can be solved in time $O(\Tso)$:
given a graph, find an orientation such that all nodes of degree at least $6$ have outdegree and indegree at least $1$.
\end{corollary}

\subsection{High-Girth: Sinkless and Sourceless Orientation}

Now we amplify the result of \Cref{corollary:deg6-sinkless-sourceless} so that we can find sinkless and sourceless orientations also in low-degree graphs---at least if we have a high-girth graph. We will now prove the following result in this section:

\begin{lemma}[High Girth, Sinkless and Sourceless]\label[lemma]{lemma:high-girth-sinkless-sourceless}
There is a constant $g$ such that the following problem can be solved in time $O(\Tso)$:
given a graph of girth at least $g$, find an orientation such that all nodes of degree at least $3$ have outdegree and indegree at least $1$.
\end{lemma}

\paragraph{Proof Overview.}
Our overall plan is as follows. Given any graph $G$ of girth at least $g$, we perform a sequence of modifications (both types of modifications are explained in detail below this proof overview) that change the degree distribution:
\begin{itemize}[noitemsep]
    \item Splitting for $d=3$: all nodes have degree 1 or exactly 3.
    \item Contraction from $d=3$ to $d'=4$: all nodes have degree 1 or at least 4.
    \item Splitting for $d=4$: all nodes have degree 1 or exactly 4.
    \item Contraction from $d=4$ to $d'=6$: all nodes have degree 1 or at least 6.
    \item Splitting for $d=6$: all nodes have degree 1 or exactly 6.
\end{itemize}
Then we apply \Cref{corollary:deg6-sinkless-sourceless} to find an orientation such that degree-$6$ nodes have outdegree and indegree at least $1$. Finally, we revert all splitting and contraction steps to recover an orientation of the original graph with the desired properties.

We will assume that $g$ is sufficiently large so that each contraction is applied to a tree-like neighborhood (in particular, contractions do neither lead to multiple parallel edges nor to self-loops). The splitting step does not create any short cycles.

\paragraph{Splitting Step.}
Given any graph and any value $d > 1$, we can apply the splitting idea from \Cref{corollary:sinkless-sourceless} to obtain a graph in which we have \emph{leaf nodes} of degree $1$ and \emph{internal nodes} of degree $d$.

The edges that join a pair of internal nodes are called \emph{internal edges}; all other edges are \emph{leaf edges}. If at any point we obtain a connected component that does not contain any internal edges, such a component is a star and we can find a valid orientation trivially in constant time. Hence let us focus on the components that contain some internal edges.

\paragraph{Contraction Step.}
Let $d' = 2d-2$. We assume that we have a graph in which all nodes are either leaf nodes or internal nodes of degree $d$, and we will show how to modify the graph so that the internal nodes have degree at least $d'$.

First, find a maximal matching $M$ of the internal edges (this is possible in time $O(\log^* n) = o(\Tso)$ with a deterministic algorithm, as we have a constant maximum degree). Then each internal node $u$ that is not matched picks arbitrarily one of its matched neighbors $v$, and adds the edge to $v$ to a set $X$. Now, $Y = M \cup X$ is a collection of internal edges that covers all internal nodes. Furthermore, each connected component in the graph induced by $Y$ has a constant diameter; it consists of an edge $e \in M$ and possibly some edges adjacent to $e$.

Now each edge $e \in M$ labels the edges of $X$ adjacent to $e$ arbitrarily with distinct labels $1, 2, \dotsc$. This way we obtain a partitioning of $Y$ in subsets $Y_0, Y_1, \dotsc, Y_k$ for some $k = O(1)$, where $Y_0 = M$ and $Y_i$, $i > 0$, consists of the edges of $X$ with label $i$.

The key observation is that each $Y_i$ is a matching. Now we do a sequence of $k+1$ edge contractions: we contract first all edges of $Y_k$, then all edges of $Y_{k-1}$, etc. For each edge that we contract, we delete the edge and identify its endpoints.

Note that all internal nodes take part in at least one edge contraction that merges a pair of internal nodes. Hence all internal nodes will have degree at least $d' = 2d-2$ after contractions. Furthermore, just before we contract the edges of $Y_i$ the edges of $Y_i$ still form a matching despite the contractions for $Y_k,\ldots,Y_{i+1}$ that we have already performed (for this property to hold it is crucial that we begin  contracting edges in $Y_k$ and not the edges in $Y_0$). Thus  we only shorten distances by a constant factor; the new graph $G'$ that we obtain can be still simulated efficiently with a distributed algorithm that runs in the original graph $G$.

\paragraph{Orientation.}
After a sequence of split and contract operations, we have a graph $H$ in which each node has degree $1$ or at least $6$. Then we apply \Cref{corollary:deg6-sinkless-sourceless} on $H$ and obtain an orientation of $H$ in which every node with degree at least $6$ has outdegree and indegree at least $1$.

\paragraph{Reverting Splits \& Contractions.}
Now we need to revert the splittings and contraction operations to turn the orientation of $H$ into an orientation of $G$. Reverting a split is trivial, but reverting a contraction needs more care to make sure that we maintain the property that all internal nodes have at least one outgoing and one incoming edge.

Consider an edge $e = \{u,v\}$ that was contracted to a single node $x$. Node $x$ is incident to at least one outgoing edge and at least one incoming edge. Revert the contraction of edge $e$ (preserving the orientations, but leaving the new edge $e$ unoriented; note that all other edges incident to $u$ or $v$ are oriented). Now if both $u$ and $v$ are happy we can orient $e$ arbitrarily. Otherwise at least one of them is unhappy; assume that $u$ is unhappy. We have the following cases:
\begin{itemize}
    \item Node $u$ is incident to only outgoing edges. Then node $v$ is incident to at least one incoming edge. Orient $e$ from $v$ to $u$. Now both $u$ and $v$ have both incoming and outgoing edges, and hence both of them are happy.
    \item Node $u$ is incident to only incoming edges. Orient $e$ from $u$ to $v$; again, both of them will be happy.
\end{itemize}
Hence we only need to invoke \Cref{corollary:deg6-sinkless-sourceless} once, in a virtual graph that can be simulated efficiently in the original network, and then do a constant number of additional operations. This completes the proof of \Cref{lemma:high-girth-sinkless-sourceless}.

\subsection{Short Cycles: Sinkless and Sourceless Orientation}\label[section]{sec:shortCycles}

The only concern that remains to prove \Cref{lemma:sinkless-sourceless} is the existence of short cycles (a special case of which is a $2$-cycle formed by a pair of parallel edges in a multigraph, and a $1$-cycle formed by a self-loop). As we will see, the existence of short cycles actually makes the problem easier to solve; only nodes that are not part of any short cycle need nontrivial computational effort.

\paragraph{Identification of Short Cycles.}
Let $g = O(1)$ be the constant from \Cref{lemma:high-girth-sinkless-sourceless}. Given a $3$-regular multigraph $G$, we first identify all cycles of length at most $g$. This is possible in time $O(1)$. Then for each cycle $C$, we assign a unique numerical identifier $i(C)$. Each cycle can for example be uniquely labelled by the sequence of node identifiers that result when starting at the highest ID node of the cycle and traversing the cycle in one of the two directions. We also pick arbitrarily an orientation $d(C)$ for the cycle.

Now let $S \subseteq E$ be the set of the edges that are involved in at least one cycle of length at most $g$, and let $X \subseteq V$ be the set of nodes involved in at least one cycle of length at most $g$. We will first orient the edges of $S$ so that all nodes in $X$ become happy, i.e., they have at least one outgoing edge and at least one incoming edge in $S$. To achieve this, we will first design a simple centralized, sequential algorithm $A$ that solves this, and then observe that we can develop an efficient distributed algorithm $A'$ that calculates in constant time the same result as what $A$ would output.

\paragraph{Centralized Algorithm.}
Algorithm $A$ proceeds as follows. We take the list of all short cycles, order them by the unique identifiers $i(C)$, and process the cycles in this order. Whenever we process some cycle $C$, we orient all edges of $C \subseteq S$ in a consistent manner, using orientation $d(C)$. While doing this, we may re-orient some edges that we had already previously oriented. Nevertheless, we make progress:
\begin{itemize}
    \item After processing cycle $C$, all nodes along $C$ are happy (regardless of whether they were already happy previously).
    \item All nodes not along $C$ that were happy before this step are also happy after this step (we did not touch any of their incident edges).
\end{itemize}
Hence after going through the list of all cycles, all edges of $S$ are oriented and all nodes of $X$ are happy.

\paragraph{Distributed Algorithm.}
The centralized algorithm is clearly inefficient for our purposes, but for each edge $e \in S$, we can directly compute what is its final orientation in the output of algorithm $A$: simply consider all cycles $C$ with $e \in C$, pick the cycle $C^*_e$ that has the largest identifier among all cycles that pass through $e$, and orient $e$ according to $d(C^*_e)$. This is easy to implement in constant time, as all cycles of interest are of constant length.

\paragraph{Remaining Nodes.}
Now all nodes of $X$ are happy. We delete all edges of $S$ and also delete all isolated nodes; this way we obtain a graph $G'$ in which all nodes have degree $1$ or $3$ and all edges are unoriented. Then we can apply \Cref{lemma:high-girth-sinkless-sourceless} to make nodes of degree $3$ happy. Finally, we put back the edges of $S$ to make all other nodes happy. This completes the proof of \Cref{lemma:sinkless-sourceless}.

%!TEX root = main.tex

\section{Degree 5: Outdegree Two}\label[section]{sec:outdegtwo}

One final piece is still missing: in \Cref{sec:shortPathsDecompositions} we used the following result but postponed its proof:

\lemmaFirstoutdegtwo*

We will simplify the problem slightly by first focusing on regular graphs. In this section we will prove the following statement:

\begin{lemma}\label[lemma]{lemma:outdeg-2}
The following problem can be solved in time $O(\log n)$ with deterministic algorithms and $O(\log \log n)$ with randomized algorithms:
given a $5$-regular multigraph, find an orientation such that all nodes have outdegree at least $2$.
\end{lemma}

The same reduction as in the proof of \Cref{corollary:sinkless-sourceless} then generalizes this result to non-regular graphs, and \Cref{lemma:firstoutdeg-2} follows directly.

\paragraph{Half-Path Decompositions.}
To prove \Cref{lemma:outdeg-2}, we start by introducing the concept of \emph{half-path decompositions}. In such a decomposition, each edge is divided in two \emph{half-edges} and we require that, for each node, exactly two incident half-edges are labeled with the color \emph{red}; all other half-edges are \emph{black}. We say that we have a \emph{decomposition with half-paths of length $k$} if the red half-edges form paths (never cycles), and each such path consists of at most $k$ half-edges but there is no requirement for black half-edges.

\begin{figure}
\centering
\includegraphics[scale=0.5]{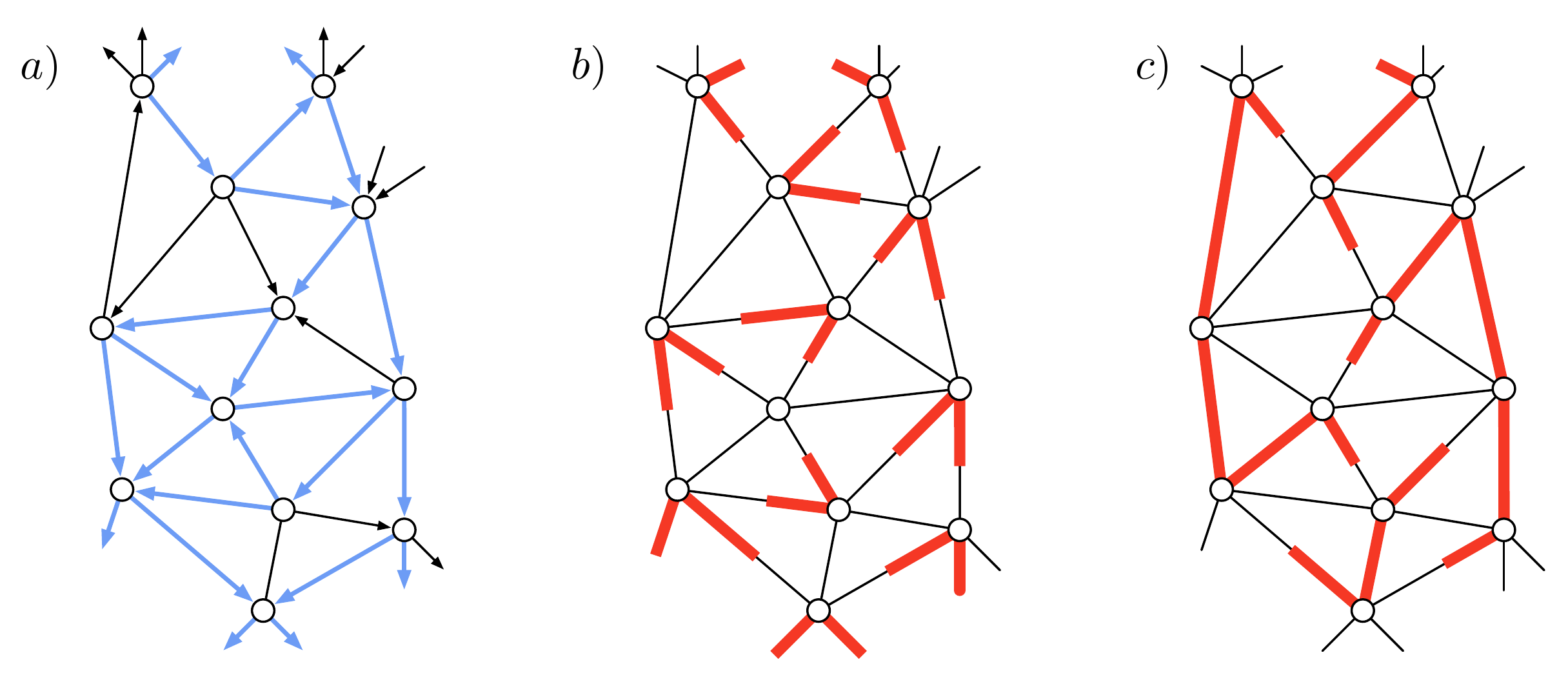}
\caption{Orientations and half-path decompositions in a $5$-regular graph. (a)~A weak $2$-orientation. Each node has selected exactly $2$ incident edges; these are indicated with a blue color. (b)~A decomposition with half-paths of length $2$. (c)~A decomposition with half-paths of length at most $8$.}\label[figure]{fig:basecase-halfpath}
\end{figure}

Half-path decompositions are closely related to weak $2$-orientations; see \Cref{fig:basecase-halfpath}. If we could find a weak $2$-orientation, each node could simply pick two outgoing edges, label their sides of these edges red, and we would have a decomposition with half-paths of length $2$. Conversely, given a decomposition with half-paths of length $2$, we could easily find a weak $2$-orientation: an edge that is half-red is oriented from red half to black half, and all other edges (which are fully black) are oriented arbitrarily.

\paragraph{Proof Idea and Intuition.}
Half-paths of length $k > 2$ can be interpreted as a relaxation of weak $2$-orientations. To find a weak $2$-orientation, we will proceed in two steps:
\begin{itemize}
    \item Find a decomposition with half-paths of length at most $8$.
    \item Use such a decomposition to find a weak $2$-orientation (in the proof of \Cref{lemma:outdeg-2}).
\end{itemize}

To get some intuition on the basic idea for computing half-path decompositions, let us first consider a simplified setting. Assume that we have a simple $5$-regular graph $G$, and assume that we are given a \emph{perfect matching} $M$. Now we could simply remove $M$, and we would be left with a $4$-regular graph $G'$. Then we could apply \Cref{lemma:weak1} to find a sinkless orientation in $G'$. Finally, we could color the half-edges of $G$ as follows:
\begin{itemize}
    \item For each edge $e \in M$, label both of its half-edges red. This contributes one red half-edge per node.
    \item For each node $v$, pick one outgoing edge in the orientation $G'$, and color the end of $v$ red and the other half black. This also contributes one half-edge per node.
\end{itemize}
We would now have a decomposition with half-paths of length $4$. Unfortunately, we cannot find a perfect matching efficiently. However, in the following lemma we will show that it is sufficient to find a \emph{maximal matching} $M$. This may result in some unmatched nodes, but the key insight is that such nodes form an independent set, and we can apply a split-and-contract trick to label those nodes; this will result in half-paths of length at most $8$.

\begin{lemma}[Half-Path Decomposition]\label[lemma]{lemma:halfpath}
The following problem can be solved in time $O(\Tso)$:
given a $5$-regular multigraph, find a decomposition with half-paths of length at most $8$.
\end{lemma}

\begin{proof}
Let $G = (V,E)$ be a $5$-regular multigraph. Let $V_L$ be the set of nodes that have at least one self-loop; for each such node, we pick one loop and add it to $L\subseteq E$.

Then find a maximal matching $M$ in the graph induced by the nodes $V \setminus V_L$; this is possible in time $O(\log^* n) = o(\Tso)$ with a deterministic algorithm, as we have a constant maximum degree. Let $V_M$ be the set of matched nodes, and let $V_U$ be the set of unmatched nodes. Note that $V_U$ is an independent set of nodes and none of these have any self-loops.

We split each node of $V_U$ arbitrarily into two parts: a node of degree $2$ and a node of degree $3$. Let $V_2$ be the set of degree-$2$ nodes, and let $V_3$ be the set of degree-$3$ nodes formed this way, and write $V_5 := V_L \cup V_M$ for all other nodes (which have degree $5$).

Note that for each $v \in V_2$, both of its neighbors are in $V_5$. Now we eliminate the nodes of $V_2$ by contracting each path of the form $V_5$--$V_2$--$V_5$; let $C$ be the set of edges that result from such contractions. We have the following setting:
\begin{itemize}[noitemsep]
    \item $C$, $L$, and $M$ are disjoint sets of edges.
    \item The endpoints of $C$, $L$, and $M$ are in $V_5$.
\end{itemize}
Now we remove the edges of $L$ and $M$. We have a multigraph $G'$ with the following sets of nodes:
\begin{itemize}[noitemsep]
    \item $V_L$: nodes of degree $3$ (they lost two endpoints when we eliminated self-loops).
    \item $V_M$: nodes of degree $4$ (they lost one endpoint when we eliminated the matching).
    \item $V_3$: nodes of degree $3$.
\end{itemize}
We find a sinkless orientation in $G'$, using e.g. \Cref{corollary:sinkless-sourceless}. Then all nodes of $V_5=V_L\cup V_M$ pick one outgoing edge and label this half-edge red. We have:
\begin{itemize}[noitemsep]
    \item Nodes of $V_L$ and $V_M$ are incident to exactly one red half-edge.
    \item Nodes of $V_3$ are not incident to any red half-edges.
    \item Each edge has at most one red half.
    \item The longest red path has length $1$.
\end{itemize}
Then we put back $M$ and label both halves of these edges red. We also put back $L$ and label exactly one endpoint of these edges red. We have:
\begin{itemize}[noitemsep]
    \item Nodes of $V_L$ and $V_M$ are incident to exactly two red half-edges.
    \item Nodes of $V_3$ are not incident to any red half-edges.
    \item The longest red path has length $4$ (an edge from $M$ plus two half-edges).
\end{itemize}
Then we revert the contractions and put back the nodes of set $V_2$. Note that each edge of $C$ had at most one red half-edge. We apply the following rules to color the new half-edges:
\begin{itemize}[noitemsep]
    \item black--black becomes black--red--red--black,
    \item red--black becomes red--red--red--black.
\end{itemize}
We obtain:
\begin{itemize}[noitemsep]
    \item Nodes of $V_L$ and $V_M$ are incident to exactly two red half-edge.
    \item Nodes of $V_3$ are not incident to any red half-edges.
    \item Nodes of $V_2$ are incident to exactly two red half-edges.
    \item The longest red path has length $8$.
\end{itemize}
Finally, we combine each pair of $u \in V_2$ and $v \in V_3$ to restore the original multigraph $G$. Here $u$ contributes two red half-edges and $v$ does not contribute any red half-edges. Overall, all nodes of $G$ are incident to exactly two red half-edges.
\end{proof}

Now we are ready to prove \Cref{lemma:outdeg-2}. Thanks to a half-path decomposition, this is straightforward. Incidentally, we get a strong $2$-orientation for free here, even though we only need a weak $2$-orientation.

\begin{proof}[Proof of \Cref{lemma:outdeg-2}]
Given a $5$-regular multigraph $G$, we first find a decomposition with half-paths of length at most $8$. Split each node into red and black parts: a degree-$2$ node incident to two red half-edges, and a degree-$3$ node incident to three black edges. Now each path of degree-$2$ nodes consists of at most $4$ nodes. We contract such paths to single edges to obtain a $3$-regular multigraph. We apply \Cref{corollary:sinkless-sourceless} to orient it (this also orients the edges that represent paths), and then undo the contractions where edges in a path are oriented according to the orientation of the edge representing the path. Now we have an oriented multigraph $G'$ in which degree-$3$ nodes have outdegree and indegree at least $1$, and degree-$2$ nodes have outdegree and indegree equal to $1$. Undo the split to get back to multigraph $G$; now each original node has outdegree and indegree at least $2$.
\end{proof}

%!TEX root = main.tex

\section{Directed and Undirected Splits}\label{sec:mainsplitting}

We are now ready to prove our main result:

\thmMainSplitting*

\begin{proof}
For both parts apply
\Cref{lemma:mainPathDecomposition}, which provides a $\big(\delta(v),\bigO(1/\eps)\big)$-path decomposition $\mathcal{P}$ with  $\delta(v)=\eps d(v)+3$ if $\eps d(v)\geq1$ and $\delta(v)=4$ otherwise.

\subparagraph{Proof of (b).}
Nodes color each path of $\mathcal{P}$ alternating with red and blue. Because the length of a path in $\mathcal{P}$ is bounded by $\bigO(1/\eps)$ this can be done in $\bigO(1/\eps)$ rounds.

Consider some node $v$ and observe that $v$ has one red and one blue edge for any path where $v$ is not a startpoint or endpoint. Thus the discrepancy of node $v$ is bounded above by $\delta(v)\leq \eps d(v)+4$.

\subparagraph{Proof of (a).}
Use \Cref{corollary:sinkless-sourceless} to compute an orientation $\pi_{\mathcal{P}}$ of $G(\mathcal{P})$ in which all nodes which have degree at least three in $G(\mathcal{P})$ have at least one incoming and one outgoing edge. Then orient paths in the original graph according to $\pi_{\mathcal{P}}$ as in the proof of \Cref{lemma:fromPathDecompToStrongOrient} and denote the resulting orientation of the edges of $G$ with  $\pi_G$.

Consider some node $v$ and observe that orienting any path that contains $v$ but where $v$ is not a startpoint or endpoint adds exactly one incoming edge and one outgoing edge for $v$.
Therefore, the discrepancy of the indegrees and outdegrees of $v$ in $\pi_{\mathcal{P}}$ bounds from above the discrepancy of the indegrees and outdegrees in $\pi_G$. The goal is to upper bound this discrepancy as desired.

Therefor let $d_{\mathcal{P}}(v)$ denote the degree of $v$ in $G(\mathcal{P})$.
If $d_{\mathcal{P}}(v)$ is at least three  then its discrepancy in $\pi_{\mathcal{P}}$ is bounded by $d_{\mathcal{P}}(v)-2$ as the algorithm from \Cref{corollary:sinkless-sourceless} provided one incoming and one outgoing edge for $v$ in $G(\mathcal{P})$. 
Furthermore we obtain that $d_{\mathcal{P}}(v)$ and $d(v)$ have the same parity because $d(v)=d_{\mathcal{P}}(v)+2x$ holds where $x$ is the number of paths that contain $v$ but where $v$ is neither a startpoint nor an endpoint.  
Thus we have the following cases.

\begin{itemize}
\item $d_{\mathcal{P}}(v)\geq 3$:
\begin{itemize}
\item$\eps d(v)\geq 1$: $v$'s  discrepancy in $\pi_G$ is bounded by $d_{\mathcal{P}}(v)-2\leq \eps d(v)+1$.
\item$\eps d(v)< 1$, $d(v)$ even: $v$'s discrepancy in $\pi_G$ is bounded by $d_{\mathcal{P}}(v)-2\leq 2$.
\item$\eps d(v)< 1$,  $d(v)$ odd:  As $d_{\mathcal{P}}(v)$ has to be odd and $3\leq d_{\mathcal{P}}(v)\leq \delta(v)=4$ holds we have $d_{\mathcal{P}}(v)= 3$. Thus $v$'s  discrepancy in $\pi_G$ is bounded by $d_{\mathcal{P}}(v)-2\leq 1$.
\end{itemize}

\item $d_{\mathcal{P}}(v)< 3$:
\begin{itemize}
\item$d(v)$ even: We have $d_{\mathcal{P}}\in \{0,2\}$ and $v$'s discrepancy in $\pi_G$ is also $0$ or $2$.
\item$d(v)$ odd: We have  $d_{\mathcal{P}} = 1$ and  $v$'s discrepancy in $\pi_G$ is also $1$.
\end{itemize}
\end{itemize}
In all cases we have that the discrepancy of node $v$ is upper bounded by $\eps d(v)+1$ if $d(v)$ is even and by $\eps d(v)+2$ if $d(v)$ is even, which proves the result.
\qedhere
\end{proof}

%!TEX root = main.tex

\section[\texorpdfstring{$((2+o(1))\Delta)$-Edge Coloring via Degree Splitting}{((2+o(1))Delta)-Edge Coloring via Degree Splitting}]{\texorpdfstring{\boldmath$((2+o(1))\Delta)$-Edge Coloring via Degree Splitting}{((2+o(1))Delta)-Edge Coloring via Degree Splitting}}\label{sec:edgeColoring}

In this section we will show how to use the undirected edge splitting algorithm to find an edge coloring:

\crlEdgeColoring*

\begin{proof}
The coloring is achieved by iterated application of the undirected splitting result of \Cref{thm:mainSplitting}. Set $\gamma = \frac{\eps}{20\log \Delta}$. In each of $h=\log \frac{\epsilon\Delta}{18}$ recursive iterations we apply the splitting of \Cref{thm:mainSplitting} with parameter $\gamma$ to each of the parts in parallel, until we reach parts with degree $O(1/\eps)$. If the maximum degree of each part before iteration $i$ is upper bounded by $\Delta_{i-1}$ the maximum degree of the parts is upper bounded by 
$\Delta_i\leq \frac{1}{2}(\Delta_{i-1}+\gamma \Delta_{i-1}+4)$ after iteration $i$. An induction on the number of iterations shows that the maximum degree of each part after iteration $h$ is upper bounded by 
\begin{align*}\left(\frac{1+\gamma}{2}\right)^h\Delta+2\sum_{i=0}^{h-1}\left(\frac{1+\gamma}{2}\right)^i\leq \left(\frac{1+\gamma}{2}\right)^h\Delta +5=:\Delta_h,\end{align*} where the last inequality follows with the geometric sum formula and with $\gamma\leq 1/10$. 
In the end, we have partitioned the edges into $2^h$ classes of maximum degree at most $\Delta_h=O(1/\eps)$. We can easily compute a $(2\Delta_h-1)$-edge coloring of each of these classes, all in parallel and with different colors, in $O(\Delta_h + \log^* n)=O(1/\eps + \log^* n)$ rounds, using the classic edge coloring algorithm of Panconesi and Rizzi~\cite{panconesi-rizzi}. Hence, we get an edge coloring of the whole graph with 
\begin{align*}2^h\cdot (2\Delta_h-1) & \leq \big((1+\gamma)^{\log \Delta}\big)2\Delta +9\cdot 2^h\leq 2e^{\eps/20}\Delta+\frac{\eps}{2}\Delta \\
& \leq \left(2+\frac{\eps}{2}\right)\Delta +\frac{\eps}{2}\Delta\leq (2+\eps)\Delta
\end{align*} colors.
Each iteration has round complexity
\begin{align*}
O\biggl(\frac{1}{\gamma}\cdot\log\frac{1}{\gamma}\cdot\log^{1.71}\log\frac{1}{\gamma}\cdot \log n\biggr) & = O\biggl(\frac{\log \Delta}{\eps}\cdot\log\frac{\log \Delta}{\eps}\cdot\log^{1.71}\log\frac{\log \Delta}{\eps}\cdot \log n \biggr) \\
& = O\biggl(\frac{\log \Delta}{\eps}\cdot\log\log \Delta\cdot \log^{1.71}\log\log \Delta \cdot \log n \biggr).
\end{align*}
 The total round complexity, over all the $\log \Delta$ iterations and the last coloring step, is
\[
\begin{split}
\log \Delta \cdot O\biggl(\frac{\log \Delta}{\eps} \cdot\log\log\Delta\cdot (\log\log\log \Delta)^{1.71} \cdot \log n \biggr) + O\biggl(1/\eps + \log^* n\biggr) = \\
\bigO\biggl(\frac{\log^2\Delta}{\epsilon} \cdot \log\log \Delta \cdot \big(\log\log\log\Delta\big)^{1.71} \cdot \log n\biggr).
\qedhere
\end{split}
\]
\end{proof}

%!TEX root = main.tex

\section[Lower Bound for Weak 2-Orientation in 4-Regular Graphs]{Lower Bound for Weak 2-Orientation \newline in 4-Regular Graphs}\label{sec:weak2orientation-lb}

We have seen that we can efficiently find, e.g., weak and strong 1-orientations in 3-regular graphs and weak and strong 2-orientations in 5-regular graphs. We will now prove that it is not possible to find weak or strong 2-orientations in 4-regular graphs efficiently.

\begin{figure}
\centering
\includegraphics[scale=0.4]{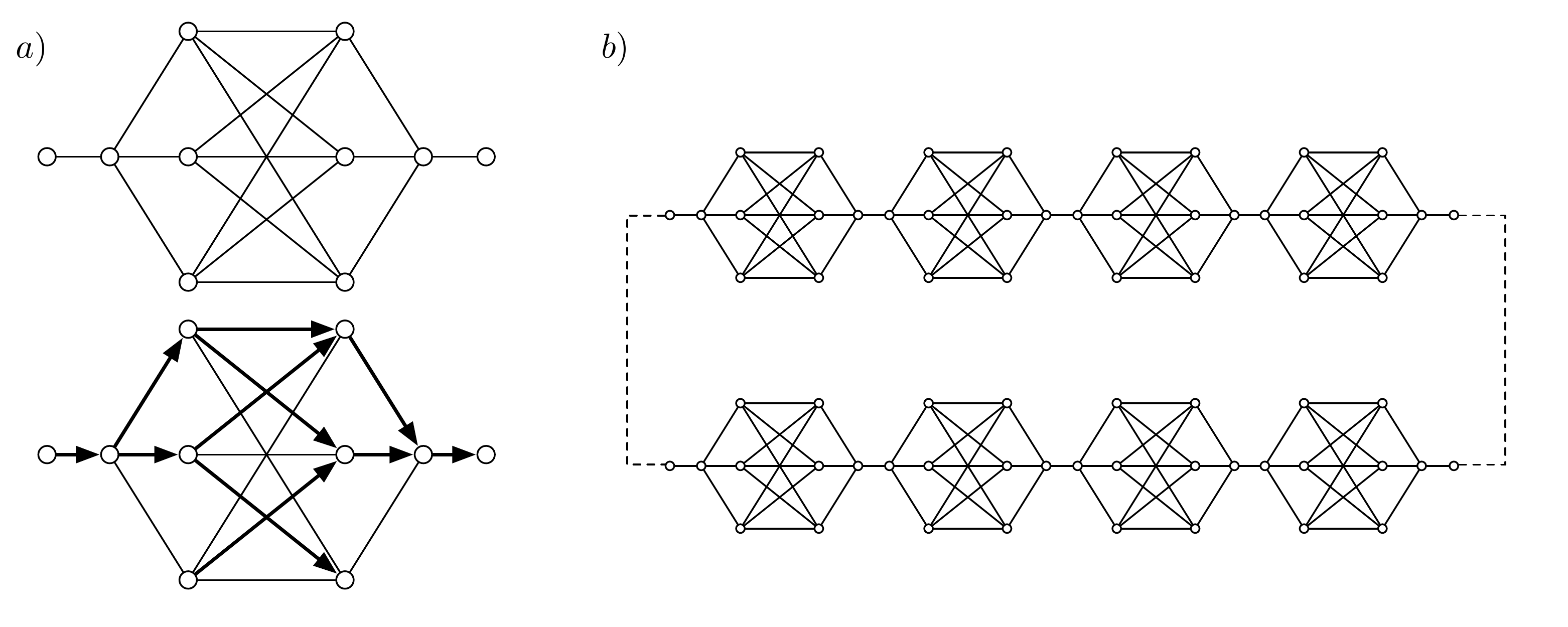}
\caption{Gadget $U$ consists of eight nodes. If $u_1$ has an incoming edge from outside the gadget, then $u_8$ must have an edge going out of the gadget. The reduction to sinkless orientation is by constructing a cycle of gadgets. The edges between the gadgets must be oriented in a consistent manner.} \label{fig:weak2orientation-gadget}
\end{figure}

\begin{theorem} \label[theorem]{thm:weak2orientation-lb}
    Weak 2-orientation in 4-regular graphs requires $\Omega(n)$ time.
\end{theorem}

\begin{proof}
  The proof is by reduction to sinkless orientation on cycles. We construct a graph consisting of constant-sized gadgets connected into a cycle such that the edges between the gadgets must be oriented consistently. This is a global problem that requires $\Omega(n)$ time.

  The gadget $U$ consists of eight nodes $V(U) = \{u_1, u_2, \dots, u_8 \}$, with $U_L = \{u_2, u_3, u_4 \}$ and $U_R = \{u_5, u_6, u_7 \}$ forming the two sides of a complete bipartite graph $K_{3,3}$. In addition, $u_1$ is connected to all nodes in $U_L$ and $u_8$ to all nodes in $U_R$.

  Now for any $n$, we construct a graph $G$ on $8n$ nodes as follows. Take $n$ copies $U_1, U_2, \dots, U_n$ of $U$, and for every $i = 1, \dots n$, connect the $i$th copy of $u_8$ (denoted by $u_{i,8}$) to $u_{i+1,1}$ modulo $n$.
  See \Cref{fig:weak2orientation-gadget} for an illustration.

  Now consider an edge $\{u_{i,8}, u_{i+1,1} \}$ and assume that it is oriented from $u_{i,8}$ to $u_{i+1,1}$. We will show that the gadget $U$ propagates orientations, that is, then we must have that $\{u_{i+1,8}, u_{i+2,1} \}$ is also oriented from $u_{i+1,8}$ to $u_{i+2,1}$. In any weak 2-orientation, $u_{i+1,1}$ must have two outgoing edges. Assume w.l.o.g. that these are to $u_{i+1,2}$ and $u_{i+1,3}$. In addition, $u_{i+1,4}$ must have an outgoing edge, giving a total of five outgoing edges from $U_{i+1, L}$ towards $U_{i+1, R}$. Therefore there must be at least two nodes in $U_{i+1, R}$ that have an outgoing edge towards $u_{i+1, 8}$, and $u_{i+1,8}$ must then have an outgoing edge toward $u_{i+2,1}$.

  Sinkless orientation requires time $\Omega(n)$ in cycles, since all edges must be oriented consistently. If weak 2-orientation could be solved in time $o(n)$ on 4-regular graphs, then nodes could virtually add gadgets $U$ between each edge and there would be an $o(n)$ time algorithm for sinkless orientation on cycles.
\end{proof}

\bibliographystyle{plainurl}
\bibliography{references}

\end{document}